\documentclass[a4paper,USenglish]{lipics-v2021}
\usepackage[utf8]{inputenc}

\usepackage{amsmath}
\usepackage{bbm}
\usepackage{listings}
\usepackage{algorithm}
\usepackage[noend]{algpseudocode}
\usepackage{newfloat}
\usepackage{float}
\usepackage{wrapfig}
\usepackage{makecell}
\usepackage[section]{placeins}
\usepackage{cleveref}

\hideLIPIcs

\makeatletter
\def\copyrightline{%
    \scriptsize
    \vtop{\hsize\textwidth\nobreakspace\par\@Copyright}%
}
\makeatother

\renewcommand{\cref}{\errmessage{In this documentclass "cref" does not work as expected. Use the capitalized version "Cref" or just plain "ref".}}
\renewcommand{\crefrange}{\errmessage{In this documentclass "cref" does not work as expected. Use the capitalized version "Cref" or just plain "ref".}}

\makeatletter
\renewcommand\theHALG@line{\thealgorithm.\arabic{ALG@line}}
\makeatother

\newenvironment{objectinterface}{\description}{\enddescription}
\newenvironment{properties}{\description}{\enddescription}

\DeclareFloatingEnvironment[fileext=eq,placement={!htp},name=Equation]{floatingequation}
\crefname{floatingequation}{equation}{equations}
\Crefname{floatingequation}{Equation}{Equations}

\definecolor{darkgray}{rgb}{0.3, 0.3, 0.3}

\newcommand{\ignore}[1]{}

\usepackage[cppcommentstyle,verysmallfont,continuouslinenumbers]{smartalgorithmic}
\NewDeclarationType{Operation}{\textsc}
\NewDeclarationType{Function}{\textsc}
\NewDeclarationType{Variable}{\mathit}
\NewDeclarationType{DistributedObject}{\textrm}
\NewDeclarationType{Type}{\textrm}
\NewDeclarationType{MessageType}{\textbf}

\usepackage[apply,nonotes]{xnotes}

\AddXNotesUser{at}{Andrei}{red}
\AddXNotesUser{lf}{Luciano}{blue}
\AddXNotesUser{pk}{Petr}{brown}

\newcommand{\Randomness}{Randomness}

\newcommand*\dif{\mathop{}\!\mathrm{d}}

\DeclareMathOperator*{\argmax}{arg\,max}

\newcommand{\Message}[1]{\ensuremath{\langle}#1\ensuremath{\rangle}}

\algblockdefx[Multiline]{Multiline}{EndMultiline}{}{}
\algtext*{EndMultiline}  

\algblockdefx[Function]{Function}{EndFunction}[2]{\textbf{function} #1(#2)}{}
\algtext*{EndFunction}  

\algblockdefx[ForEach]{ForEach}{EndForEach}[1]{\textbf{for each} #1 \textbf{do}}{}
\algtext*{EndForEach}  

\algblockdefx[Repeat]{Repeat}{EndRepeat}[1]{\textbf{repeat} #1}{}
\algtext*{EndRepeat}

\algblockdefx[Procedure]{Procedure}{EndProcedure}[2]{\textbf{procedure} #1(#2)}{}
\algtext*{EndProcedure}  

\algblockdefx[Operation]{Operation}{EndOperation}[2]{\textbf{operation} #1(#2)}{}
\algtext*{EndOperation}  

\algblockdefx[OverrideOperation]{OverrideOperation}{EndOperation}[2]{\textbf{override operation} #1(#2)}{}
\algtext*{EndOperation}  

\algblockdefx[Upon]{Upon}{EndHandler}[1]{\textbf{upon} #1}{}
\algtext*{EndHandler}  

\algblockdefx[UponEvent]{UponEvent}{EndUponEvent}[1]{\textbf{upon event} #1}{}
\algtext*{EndUponEvent}  

\algblockdefx[UponReceiving]{UponReceiving}{EndUponReceiving}[2]{\textbf{upon receiving} \Message{#1} \textbf{from} #2}{}
\algtext*{EndUponReceiving}  

\algblockdefx[UponConsistentDeliver]{UponConsistentDeliver}{EndUponConsistentDeliver}[2]{\textbf{upon consistently delivering} $\langle$#1$\rangle$ \textbf{from} #2}{}
\algtext*{EndUponConsistentDeliver}  

\algblockdefx[ProcessState]{ProcessState}{EndProcessState}{\textbf{State:}}{}
\algtext*{EndProcessState}  

\algblockdefx[DistributedObjects]{DistributedObjects}{EndDistributedObjects}{\textbf{Distributed objects:}}{}
\algtext*{EndDistributedObjects}  

\algblockdefx[FunctionDeclaration]{FunctionDeclaration}{EndFunctionDeclaration}{\textbf{Functions:}}{}
\algtext*{EndFunctionDeclaration}

\algnewcommand{\Returns}[1]{\textbf{returns} {#1}}

\algnewcommand{\WaitFor}{\textbf{wait for} }
\algnewcommand{\WaitUntil}{\textbf{wait until} }
\algnewcommand{\Send}[2]{\textbf{send} $\langle$#1$\rangle$ \textbf{to} #2}
\algnewcommand{\IfThenElse}[3]{
    \algorithmicif\ #1\ \algorithmicthen\ #2\ \algorithmicelse\ #3}
\algnewcommand{\IfThen}[2]{
    \algorithmicif\ #1\ \algorithmicthen\ #2}

\algnewcommand{\algspace}{\Statex}

\algnewcommand{\TriggerEvent}[1]{\textbf{trigger event } #1}

\algnewcommand{\Parameters}[1]{\State \textbf{Instance parameters: } #1}

\DeclareGlobalFunction{CoinToss}

\newcommand{\pRandomness}{Randomness}
\newcommand{\pTermination}{Termination}
\newcommand{\pAgreement}{Agreement}

\newcommand{\pOneProcessRandomness}{One Process Randomness}

\newcommand{\pConsistency}{Consistency}
\newcommand{\pValidity}{Validity}

\newcommand{\pTotality}{Totality}
\newcommand{\pNotificationTotality}{Notification Totality}
\newcommand{\pAcceptValidity}{Accept Validity}
\newcommand{\pAcceptTotality}{Accept Totality}

\newcommand{\pAssignmentTermination}{Assignment Termination}
\newcommand{\pUnpredictability}{Unpredictability}
\newcommand{\pRevealTermination}{Reveal Termination}
\newcommand{\pBindingCommonCore}{Binding Common Core}
\newcommand{\pApproximateEpsConsistency}{Approximate $\epsilon$-Consistency}
\newcommand{\pProbabilistcDeltaConsistency}{Probabilistic $\delta$-Consistency}

\newcommand{\pApproximateAgreementConsistency}{Approximate $\aaEpsilon$-consistency}

\DeclareGlobalType[Integer]{integer}
\DeclareGlobalType[Set]{set of}
\DeclareGlobalType[Boolean]{boolean}

\DeclareGlobalVariable*[AllProc]{\ensuremath{[n]}}
\DeclareGlobalVariable*[Domain]{\ensuremath{D}}
\DeclareGlobalVariable*[DomainRange]{\ensuremath{[0..\Domain{-}1]}}
\DeclareGlobalVariable*[True]{\textit{true}}
\DeclareGlobalVariable*[False]{\textit{false}}

\DeclareGlobalDistributedObject{AVSS}
\DeclareGlobalDistributedObject[AVSSOLD]{AVSS\_OLD}
\DeclareGlobalFunction{ShareSecret}
\DeclareGlobalFunction{SharingComplete}
\DeclareGlobalFunction{VerifyCompletion}
\DeclareGlobalFunction{EnableRetrieve}
\DeclareGlobalFunction{Retrieve}
\newcommand{\pAVSSValidity}{Validity}
\newcommand{\pRetrieveTermination}{Retrieve Termination}
\newcommand{\pAVSSBinding}{Binding}
\newcommand{\pAVSSSecrecy}{Secrecy}

\DeclareGlobalDistributedObject{APVSS}
\DeclareGlobalFunction{CreateScriptShare}
\DeclareGlobalFunction{VerifyScriptShare}
\newcommand{\pAPVSSSecrecy}{Secrecy}

\DeclareGlobalDistributedObject{BCB}
\DeclareGlobalDistributedObject{BRB}
\DeclareGlobalFunction{Broadcast}
\DeclareGlobalFunction{Deliver}

\DeclareGlobalDistributedObject{RSD}
\DeclareGlobalFunction[StartRSD]{Start}
\DeclareGlobalFunction{ValueAssigned}
\DeclareGlobalFunction{RetrieveValues}

\DeclareGlobalFunction{Toss}

\DeclareGlobalFunction{Propose}
\DeclareGlobalFunction{Decide}

\DeclareGlobalDistributedObject{Gather}
\DeclareGlobalFunction[StartGather]{Start}
\DeclareGlobalFunction{DeliverSet}

\DeclareGlobalDistributedObject[BAA]{BAA}
\DeclareGlobalFunction[RunBAA]{Run}

\DeclareGlobalDistributedObject{TicketDraw}
\DeclareGlobalDistributedObject{ValueDraw}

\DeclareGlobalFunction{RandomInt}
\DeclareGlobalFunction{Sign}
\DeclareGlobalFunction{Verify} 
\DeclareGlobalFunction{VerifySignature}
\DeclareGlobalFunction{GenerateSymmetricKey}
\DeclareGlobalFunction{Encrypt}
\DeclareGlobalFunction{Decrypt}
\DeclareGlobalFunction*[Validity]{\mathcal{P}}
\DeclareGlobalFunction{GatherAccept}
\DeclareGlobalFunction{VPEAccept}
\DeclareGlobalFunction{Calibrate}

\DeclareGlobalCommand{\dist}[1]{\ensuremath{d_{#1}}}

\DeclareGlobalVariable[vvalue]{value}  
\DeclareGlobalVariable*[CommonCore]{\mathcal{S}^*}
\DeclareGlobalVariable*[aaEpsilon]{\ensuremath{\tilde\epsilon}}
\DeclareGlobalVariable*[bigDomain]{\ensuremath{D}}
\DeclareGlobalVariable{scriptShare}
\DeclareGlobalVariable{shares}
\DeclareGlobalVariable{scriptShares}
\DeclareGlobalVariable{script}
\DeclareGlobalVariable[domainFactor]{k}

\DeclareUnicodeCharacter{2212}{$-$}

\newcommand{\ApproxCoinEpsilon}{\epsilon/f}
\newcommand{\ApproxCoinRounds}{\log_2 f + \log_2(1 / \epsilon)}

\newcommand{\ProbcoinRoundsWithCalibrationWithQ}{5 + \lceil\log_2(1/Q) + \log_2(\log_2(1/Q))\rceil}

\newcommand{\ProbcoinRoundsWithoutCalibrationWithQ}{3 + \lceil\log_2(n) + \log_2(1/Q)\rceil}

\newcommand{\ProbcoinRoundsWithoutCalibrationWithDelta}{3 + \lceil\log_2(n) + \log_2\left(\frac{1}{1-\delta}\right)\rceil}

\newcommand{\ProbcoinRoundsWithCalibrationWithDelta}{5 + \left\lceil\log_2\left(\frac{1}{1 - \delta}\right) + \log_2\left(\log_2\left(\frac{1}{1 - \delta}\right)\right)\right\rceil}

\newcommand{\ProbcoinCalibrationMinimumNWithQ}{\frac{3 \ln(2/Q)}{2}}

\newcommand{\myparagraph}[1]{\subparagraph*{\textbf{#1.}}}

\title{Distributed Randomness from Approximate Agreement}

\date{\today}



\author{Luciano Freitas}{LTCI, T\'el\'ecom Paris, Institut Polytechnique de Paris}{lfreitas@telecom-paris.fr}{}{}
\author{Petr Kuznetsov}{LTCI, T\'el\'ecom Paris, Institut Polytechnique de Paris}{petr.kuznetsov@telecom-paris.fr}{}{}
\author{Andrei Tonkikh}{LTCI, T\'el\'ecom Paris, Institut Polytechnique de Paris}{andrei.tonkikh@telecom-paris.fr}{}{}

\authorrunning{Luciano Freitas, Petr Kuznetsov, Andrei Tonkikh}

\Copyright{L. Freitas, P. Kuznetsov, A. Tonkikh}

\nolinenumbers

\begin{document}

\maketitle

\begin{abstract}
Randomisation is a critical tool in designing distributed systems.
The \emph{common coin} primitive, enabling the system members to agree on an unpredictable random number, has proven to be particularly useful.
We observe, however, that it is impossible to implement a truly random common coin protocol in a fault-prone asynchronous system.

To circumvent this impossibility, we introduce two relaxations of the perfect common coin: (1)~\emph{approximate common coin} generating random numbers that are \emph{close} to each other; and (2)~\emph{Monte Carlo common coin} generating a common random number with an arbitrarily small, but non-zero, probability of failure.
Building atop the \emph{approximate agreement} primitive, we obtain efficient asynchronous implementations of the two abstractions, tolerating up to one third of Byzantine processes.
Our protocols do not assume trusted setup or public key infrastructure and converge to the perfect coin exponentially fast in the protocol running time.

By plugging one of our protocols for \emph{Monte Carlo common coin} in a well-known consensus algorithm, we manage to get a \emph{binary Byzantine agreement} protocol with $O(n^3 \log n)$ communication complexity, resilient against an adaptive adversary, and tolerating the optimal number $f<n/3$ of failures without trusted setup or PKI.
To the best of our knowledge, the best communication complexity for binary Byzantine agreement achieved so far in this setting is $O(n^4)$.
We also show how the \emph{approximate common coin}, combined with a variant of Gray code, can be used to solve an interesting problem of Intersecting Random Subsets, which we introduce in this paper.
\end{abstract}

\section{Introduction}

Generating randomness in distributed systems is an essential part of many protocols, such as Byzantine Agreement~\cite{BenOr83}, Distributed Key Generation~\cite{gennaro1999secure}  or Leader Election~\cite{mostefaoui2001leader}. 
%
Any application that needs an unpredictable or unbiased result will most likely rely on randomness. 
Although sometimes local sources of randomness are enough for some protocols~\cite{mostefaoui2019new}, having access to a common random number can guarantee faster termination~\cite{aguilera2012correctness}.
Producing a common unpredictable random number has been extensively studied in the literature on cryptography and distributed systems under the names of \emph{random beacon}, distributed (multi-party) \emph{random number generation} or \emph{common coin} (even if the result is not binary). 
In essence, these protocols ensure:

\begin{properties}
    \item[\pTermination:] every correct process eventually outputs some value;

    \item[\pAgreement:] no two correct processes output different values;
    
    \item[\pRandomness:] the value output by a correct process must be uniformly distributed over some domain $\mathcal{D}$, $|\mathcal{D}| \ge 2$.
\end{properties}

We call a protocol that ensures the three properties  
(\emph{\pTermination}, \emph{\pAgreement}, and \emph{\pRandomness}) a \textbf{perfect common coin}. %
There are many message-passing protocols without trusted setup that implement a perfect common coin in the presence of Byzantine adversary \cite{syta2017scalable,cascudo2017scrape,DBLP:journals/corr/abs-1805-04548,cachin2005random,bunz2017proofs,krasnoselskii2020no,de2021randsolomon}.
These protocols are either \emph{synchronous}, meaning that every message sent by a correct process is delivered within a certain (known \emph{a priori}) bound of time, or \emph{partially-synchronous}, meaning that such a bound exists but is unknown.

In contrast, one can also consider an  \emph{asynchronous} system, where no bounds on communication delays can be assumed. 
In a seminal work of Fischer, Lynch, and Paterson, it has been shown that the problem of \emph{consensus} has no  asynchronous fault-tolerant solutions~\cite{FLP85}. 
As we show in \Cref{app:impossibility}, this impossibility also holds for perfect common coins:
no algorithm can implement a perfect common coin in a message-passing asynchronous system where at least one process might crash.
Note that
this statement cannot be proven by a simple black-box reduction from consensus to a perfect common coin and a reference to FLP~\cite{FLP85}.
Indeed, if such a reduction existed, the resulting protocol would have to always terminate in a bounded number of steps,
even with unfavourable outputs of the black-box common coin. 
Hence, if we were to replace the common coin protocol by a protocol that always returns $0$, it would still provide termination as well as all other properties of consensus, violating~\cite{FLP85}.

%
\atrev{%
Note that this impossibility applies even to systems with \emph{trusted setup}, such as the one assumed in~\cite{cachin2005random}.
Such protocols typically do not satisfy the {\Randomness} property of a perfect common coin.
The outputs of these protocols follow deterministically from the information received by the processes during the setup.}

In light of the impossibility of a perfect coin, one might look for \emph{relaxed} versions of the common-coin problem that allow asynchronous fault-tolerant solutions.
For example, one can relax the Agreement property by only requiring the output to be common with some constant probability, which results in an abstraction sometimes called 
\emph{weak common coin}\footnote{\atadd{Many }weak common coin protocols such as the one in~\cite{CanettiRabin93} also relax {\pRandomness}.}
or \emph{$\delta$-matching common coin}. 
In this paper, we call this abstraction a \textbf{probabilistic common coin}, in order to avoid confusion with other relaxations we introduce. 
More precisely, probabilistic common coins replace the Agreement property above with \emph{probabilistic $\delta$-consistency}.

\begin{properties}
    \item[\pProbabilistcDeltaConsistency:] 
    with probability at least $\delta$, no two correct processes output different values.
\end{properties}

We also introduce the concept of a \emph{Monte Carlo common coin} which is a probabilistic common coin whose success rate $\delta$ can be parameterized as follows:
%
the more rounds of the protocol are executed, the more reliable the outcome is. In our case $\delta$ starts at $\frac{2}{3}$ in the first round of the protocol and converges to $1$ at an exponential rate in the number of rounds.
%
\atrev{In most probabilistic common coins, $\delta$ could be increased by decreasing the resilience level (the allowed fraction of Byzantine processes).}
\atrev{However, to the best of our knowledge, this paper is the first to present an implementation of a Monte Carlo common coin that can achieve arbitrarily small (but non-zero) $\delta$ by increasing the running time of the protocol (\Cref{sec:probcoin-from-approxcoin,sec:probcoin}).}
%

We also propose a novel, alternative relaxation of Agreement: instead of ensuring that the \emph{same} output is produced (with some probability), we may require that the produced outputs are \emph{close} to each other according to some metric.
For this variant, we have to also slightly relax randomness so that only one correct process is guaranteed to obtain a truly random value.
More precisely, assume a discrete range of possible outputs $\DomainRange$, and let $\dist{q}(x, y)$ denote the distance between $x$ and $y$ in the\atadd{ algebraic} ring $\mathbb{Z}_{q}$\footnote{I.e., $d(x,y) = \min\{|x-y|, q-|x-y|\}.$}. 
The \textbf{Approximate common coin} abstraction then satisfies {\pTermination} and the following two properties:
\begin{properties}
    \item[\pApproximateEpsConsistency:] 
    if one correct process outputs value $x$ and another correct process outputs $y$, then $\dist{D}(x,y) \le \lceil \epsilon D \rceil$, for a given parameter $\epsilon \in (0, 1]$;
    
    \item[\pOneProcessRandomness:]
    the value output by \emph{at least one} correct process must be uniformly distributed over the domain $\DomainRange$.\footnote{With a small modification to our protocol, we can easily achieve $(f{+}1)$-Process Randomness instead on {\pOneProcessRandomness}. However, we do not know if it possible to guarantee that all outputs of correct processes are random without relaxing other properties.}
\end{properties}

%
%

Our implementations of Monte Carlo and Approximate common coins build upon the abstraction of  \emph{Approximate Agreement}~\cite{DLPSW86}.
It appears that the abstraction perfectly matches the requirements exposed by our relaxed common-coin definitions:
it naturally grasps the notion of outputs being close where the precision can be related to the execution time.  
Building atop existing asynchronous Byzantine fault-tolerant implementations~\cite{DLPSW86,abraham2004optimal}, 
we introduce and discuss an efficient implementation of the \emph{bundled\atremove{ binary}} version of this abstraction which is, intuitively, equivalent to $n$ parallel instances of Approximate Agreement\atremove{ with binary ($0$ or $1$) inputs}\atadd{, but is much more efficient}.

We discuss two applications of our protocols.
First, we observe that our Monte Carlo\atadd{ common} coin can be plugged into\atreplace{ any protocol that solves Byzantine agreement using a probabilistic common coin}{ many existing Byzantine agreement protocols}~\cite{bracha-brb,CanettiRabin93,crain2020algorithms,mostefaoui2015signature}.
This helps us to obtain a \emph{binary Byzantine agreement} protocol with $O(n^3 \lambda \log n)$ communication complexity, where $\lambda$ is the security parameter.
\atrev{The protocol exhibits optimal resilience of $f<n/3$, tolerates adaptive adversary, and assumes no trusted setup or PKI.}
In this \atreplace{model}{setting}, the best prior protocols for binary Byzantine agreement we are aware of have communication complexity of $O(n^4\lambda)$~\cite{kogias2020adkg,abraham2021reaching}.

We also introduce \emph{Intersecting Random Subsets}, a new problem that can be used to asynchronously choose random committees with large intersections.
Using elements of coding theory, namely Gray Codes~\cite{gray-codes,coding-theory-textbook}, we show how our Approximate common coin can be used to solve this problem without additional communication overhead.

We present our model definitions in \Cref{sec:system} and describe the building blocks used in our constructions in \Cref{sec:building}. %
We describe our protocols in \Cref{sec:approxcoin,sec:probcoin-from-approxcoin,sec:probcoin}, including implementation details and complexity analysis.
In~\Cref{sec:applications}, we describe applications of our abstractions: binary Byzantine agreement and Intersecting Random Subsets.
In~\Cref{sec:related}, we overview the related work and   in~\Cref{sec:conclusion}, we conclude the paper. 

\pkrev{%
For completeness, we delegate all necessary complementary material to the appendix. 
In~\Cref{app:impossibility}, we prove the impossibility of an asynchronous perfect common coin.
In~\Cref{app:approxcoin} and~\Cref{proof:prob}, we prove correctness of our approximate and Monte Carlo common coins, respectively.
\Cref{sec:rsd-impl} presents two implementations of Random Secret Draw, one of the major building blocks of our common coins. 
\Cref{app:cr93} gives a \atreplace{modern}{modular} formulation of a common coin proposed by Canetti and Rabin in 1993~\cite{CanettiRabin93}. 
Finally, we show how to build the codewords necessary for our Intersecting Random Subsets \atreplace{applications}{application} in~\Cref{app:committee-elections-concrete-code}.}

\section{System Model}
\label{sec:system}
We consider a system of $n$ processes able to communicate using reliable communication channels. Among the participants, \pkadd{at most} $f < \frac{n}{3}$ are Byzantine and might display arbitrary behaviour.

%
\pkrev{We assume the \emph{adaptive} adversarial model: up to $f$ Byzantine processes are chosen by the adversary depending on the execution.
A non-Byzantine process is called \emph{correct}.
The communication complexity of our baseline protocols can be improved by a factor of $n$ using  Aggregatable Publicly Verifiable Secret Sharing (APVSS)~\cite{gurkan2021aggregatable}\atremove{ (instead of $n$ instances of Asynchronous Verifiable Secret Sharing~\cite{avss-cachin-2002,quadratic-brb})}.
However, as we are not aware of APVSS implementations that are secure against the adaptive adversary, the improved protocols can only be proved correct in the presence of the static one.}
%

%
The adversary can control the time the messages sent by correct processes take to arrive, as well as reorder them.
However, it cannot drop a message sent by a correct process unless it corrupts this process before the message has arrived.

We assume that each process has access to a local random number generator that can be accessed as follows:

\begin{objectinterface}
    \item $\RandomInt(\Domain)$: produces a uniformly distributed random integer number in the range $\DomainRange$.
\end{objectinterface}



The proposed protocols as well as some of the building blocks rely on the use of cryptographic hash functions. The hash of an arbitrary string $s$ is denoted $H(s)$\atremove{ in the text} and has length $\lambda$ \pkrev{that we call the \emph{security parameter}}. It is computationally infeasible to find two strings $s \neq s'$ such that $H(s) = H(s')$, as well as inverting a hash without knowing which input was used a priori.

We assume a computationally bounded adversary, so that it is incapable of breaking cryptographic primitives with all but negligible probability.
%
\atrev{However, since such a probability exists, we allow the properties of all our protocols as well as all building blocks to be violated with a negligible in $\lambda$ probability.}
 
\section{Building Blocks}
\label{sec:building}

Our protocols make use of a wide range of building blocks.
None of them are completely new, but some of them are modified according to our needs.
In particular, we introduce the \emph{Random Secret Draw} abstraction inspired by the ideas from~\cite{CanettiRabin93} and~\cite{abraham2021reaching} (\Cref{sec:rsd,sec:rsd-impl}). We also provide a
\emph{bundled} version of the \emph{Approximate Agreement}~\cite{DLPSW86,abraham2004optimal} abstraction (\Cref{sec:aa}).
In addition, we use \emph{Byzantine Reliable Broadcast}~\cite{bracha-brb,quadratic-brb} (\Cref{sec:brb}), \emph{Asynchronous Verifiable Secret Sharing}~\cite{avss-cachin-2002,quadratic-brb} (\Cref{sec:avss}), \emph{Aggregatable Publicly Verifiable Secret Sharing}~\cite{gurkan2021aggregatable} (\Cref{subsubsec:apvss}), 
\atrev{and \emph{Gather}~\cite{CanettiRabin93,gather-blog} (\Cref{sec:gather})}.

\subsection{Byzantine Reliable Broadcast} \label{sec:brb}

\emph{Byzantine Reliable Broadcast (BRB)}~\cite{bracha-brb} allows a designated leader to communicate a single message to all processes in such a way that, if any correct process delivers a message, then every other correct process eventually delivers exactly the same message (even if the leader is Byzantine).
More precisely, a BRB protocol must satisfy the following properties:
\begin{properties}
    \item[\pValidity:] if the leader is correct and it broadcasts message $m$, then every correct process will eventually deliver $m$;

    \item[\pConsistency:] if two correct processes $j$ and $k$ deliver messages $m_j$ and $m_k$, then $m_j = m_k$.
    
    \item[\pTotality:] if a correct process $j$ delivers some message $m$, then eventually all correct processes will deliver $m$.

\end{properties}


The performance of reliable broadcast is of crucial importance to\atremove{ all} our protocols.
We believe that the BRB implementation recently proposed by Das, Xiang, and Ren~\cite{quadratic-brb} will be the most suitable option.
It has total communication complexity of just $O(n |M| + n^2 \lambda)$, where $|M|$ is the size of the message and $\lambda$ is the 
security parameter, total message complexity of $O(n^2)$, and the latency of $3$ message delays in case of a correct leader and $4$ message delays in case of a Byzantine leader.

In this paper, we always use BRB\atremove{ and VBRB} in groups of $n$ instances, with each process being the leader of one.
We use the following notation:
\begin{objectinterface}
    \item[$\BRB_i.\Broadcast(m)$:] \atrev{allows process $i$ to broadcast a message in an instance of BRB where $i$ is the leader;}
    
    \item[$\BRB_i.\Deliver(m)$:] \atrev{an event indicating that message $m$ from process $i$ has been delivered.}
\end{objectinterface}

\subsection{Asynchronous Verifiable Secret Sharing} \label{sec:avss}

\emph{Asynchronous Verifiable Secret Sharing (AVSS)}~\cite{avss-cachin-2002} allows a process to securely share information with other participants and to keep its contents secret until the moment a threshold of participants agree to open it.

In our protocols, AVSS is used with the following interface:

\begin{objectinterface}
    \item[$\AVSS_i.\ShareSecret(x)$:] allows process $i$ to share a secret $x$ among the participants;
    
    \item[$\AVSS_i.\SharingComplete()$:] an event issued when a secret is correctly shared by process $i$;
    

    \item[$\AVSS_i.\EnableRetrieve()$:] enables responses to retrieval requests;

    \item[$\AVSS_i.\Retrieve()$:] returns $x$ if it was previously shared and all correct processes invoked $\AVSS_i.\EnableRetrieve()$.
\end{objectinterface}

An AVSS implementation must satisfy the following properties:
\begin{properties}
    \item[\pAVSSValidity:] if a correct process $i$ invokes $\AVSS_i.\ShareSecret(x)$, then every correct process eventually receives the $\AVSS_i.\SharingComplete()$ event and no value other than $x$ can be returned from the $\AVSS_i.\Retrieve()$ operation invoked by a correct process;

    \item[\pNotificationTotality:] if one correct process receives the $\AVSS_i.\SharingComplete()$ event, then every correct process eventually receives it;

    \item[\pRetrieveTermination:] if all correct processes invoke $\AVSS_i.\EnableRetrieve()$ and any correct process invokes $\AVSS_i.\Retrieve()$, then this operation will eventually terminate and the process will obtain the shared secret;

    \item[\pAVSSBinding:] if some correct process receives the $\AVSS_i.\SharingComplete()$ notification, then there exists a fixed secret $x$ such that no value other than $x$ can be returned from the $\AVSS_i.\Retrieve()$ operation invoked by a correct process;
    
    \item[\pAVSSSecrecy:] if process $i$ is correct and no correct process invoked $\AVSS_i.\EnableRetrieve()$, then the adversary has no information about the secret shared by $i$.
\end{properties}

Das, Xiang, and Ren~\cite{quadratic-brb} proposed an {\AVSS} protocol with quadratic communication complexity, constant latency, and without assuming trusted setup.
Notice that in order to secretly share a long string $s$, it is better to follow the method proposed in~\cite{krawczyk1993secret}: encrypt $s$ using a much shorter secret key \emph{sym}, reliably broadcast the encrypted value $\{s\}_{\mathit{sym}}$
and then perform secret sharing of the key \emph{sym}.
Thus, we shall assume that the total communication complexity of secret sharing of string $s$ is $O(n|s| + n^2\lambda)$.

%

\subsection{Random Secret Draw}
\label{sec:rsd}

One of the key ideas of the weak common coin protocol of Canetti and Rabin~\cite{CanettiRabin93}
is to \emph{assign} each process a random number in a given domain $\DomainRange$ in such a way that:
\begin{properties}
    \item[\pAssignmentTermination:] if a correct process $i$ participates, then it is eventually assigned a value. Moreover, everyone will eventually receive a notification that $i$ has been assigned a random value;
    
    \item[\pNotificationTotality:] if process $i$ receives a notification that some process $j$ has been assigned a value, then every correct process will eventually receive such a notification;

    \item[\pRandomness:] the assigned numbers are independent and distributed uniformly over the domain $\DomainRange$.
    The distribution of the value assigned to process $j$ cannot be affected by the adversary even if $j$ itself is Byzantine;

    \item[\pUnpredictability:] until at least one correct process agrees to reveal the assigned values, the value assigned to each process $j$ remains secret, even to process $j$ itself;
    
    \item[\pRevealTermination:] if all correct processes want to reveal the assigned values, then they will eventually succeed.
\end{properties}

Although this idea has been widely used as part of the implementation of asynchronous consensus protocols, to the best of our knowledge, it was never considered a separate primitive and assigned a name.
Hence, we shall call it \emph{Random Secret Draw (RSD)}.

This abstraction resembles a well known concept of a Verifiable Random Function (VRF)~\cite{vrf}.
However, the important difference is that process $j$ itself cannot know the value it is assigned until the reveal phase.
Hence, a Byzantine process cannot choose whether it wants to participate or not based on the random value it is assigned.
Moreover, unlike Random Secret Draw, VRF schemes typically require a seed chosen at random \emph{after} the process chose the public key for its  pseudo-random function.
In fact, a variant of RSD has been recently used to generate such seeds~\cite{gao2021efficient}.

We use the following interface for the RSD abstraction:

\begin{objectinterface}
    \item[$\RSD.\StartRSD()$:] allows a process to start participating in RSD and, eventually, to be assigned a random number. We assume that this function is non-blocking, i.e., that an invocation of this function terminates after $0$ message delays;
    
    \item[$\RSD.\EnableRetrieve():$] used by the processes to start participating in the process of re\-cons\-truc\-ting the assigned values;
    
    \item[$\RSD.\RetrieveValues(S)$:] returns a map from the ids of processes in the set $S$ to the assigned values if all correct processes\atremove[They can invoke it after]{ \pkadd{previously}} invoked $\RSD.\EnableRetrieve()$ and all processes in $S$ have been assigned some values. 
\end{objectinterface}

The original RSD implementation by Canetti and Rabin~\cite{CanettiRabin93} used $n^2$ instances of AVSS.
To the best of our knowledge, to this day, there is no known AVSS protocol that would allow to do it with less than $\Omega(n^4)$ bits of communication in total.
We, therefore, give two possible implementations. The first one is secure against an adaptive adversary and does not rely on PKI, while the second one uses the implementation from~\cite{abraham2021reaching} that relies of \emph{Aggregatable Publicly Verifiable Secret Sharing} (described in \Cref{subsubsec:apvss}) instead of AVSS.
\atrev{While saving a linear factor in communication complexity, this solution lacks security against adaptive adversary and requires PKI.}
%
%
\atrev{Since both~\cite{CanettiRabin93} and~\cite{abraham2021reaching} did not considered RSD as a separate abstraction and did not provide separate pseudocode for it, we present both RSD implementations in \Cref{sec:rsd-impl}.}

\subsection{Gather}
\label{sec:gather}

\DeclareSectionVariable*[GatherSet]{\mathcal{S}}

Yet another important contribution made by Canetti and Rabin in their weak common coin construction~\cite{CanettiRabin93} is a multi-broadcast protocol that has been recently given the name \emph{Gather}~\cite{gather-blog,abraham2021reaching}.
In this protocol, every process starts by broadcasting a single message through Byzantine Reliable Broadcast.
The processes then do a few more rounds of message exchanges and, in the end, each participant $i$ outputs a set of process ids $\GatherSet_i$ such that for all $j \in \GatherSet_i$: $i$ has received the message of $j$ through reliable broadcast.\footnote{In~\cite{gather-blog} and~\cite{abraham2021reaching}, Gather returns a set of pairs $(id, \vvalue)$. However, for our purposes, working with sets of ids is more convenient. The values will be delivered through normal $\BRB.\Deliver$ event.}
Moreover, the sets output by correct processes satisfy a strong intersection property:

\begin{properties}
    \item[\pBindingCommonCore:] There exists a set $\CommonCore$ of process ids of size at least $n-f$, called the \emph{common core}, such that for every correct process $i$: $\CommonCore \subseteq \GatherSet_i$. Moreover, once the first correct process outputs, $\CommonCore$ is fixed and the adversary cannot manipulate it anymore.
\end{properties}

The fact that the adversary cannot affect the \atreplace{core set}{common core} once a single correct process outputs will be important in our protocols.
%
\atrev{The adversary should not be able to choose the common core based on the generated random numbers after some of the correct processes invoked $\EnableRetrieve$.}


We slightly generalize the interface of Gather by using it in conjunction with BRB, but also with other similar primitives (in particular, AVSS and RSD) and their combinations.
When a process invokes $\Gather$, it passes to it an arbitrary callable function $\GatherAccept$ that takes a process id $j$ and returns $\True$ if the message from this process is considered to be delivered (not necessarily through BRB).
\atrev{We assume that Gather exports the following interface}:

\begin{objectinterface}
    \item[$\Gather.\StartGather(\GatherAccept)$:] allows a process to start participating in the Gather protocol;
    \item[$\Gather.\DeliverSet(S)$:] \atrev{provides the output of the Gather protocol.}
\end{objectinterface}

In order for the protocol to terminate, the $\GatherAccept$ function has to satisfy properties similar to those of reliable broadcast.
\begin{properties}
    \item[\pAcceptValidity:] if a correct process $i$ invoked $\Gather.\StartGather$, then for every correct process $j$, $\GatherAccept(i)$ invoked by process $j$ must eventually return $\True$\atadd{. Moreover, for all $i$, once $\GatherAccept(i)$ returned to $\True$ to some correct process, it must never switch back to $\False$};
    

    \item[\pAcceptTotality:] if $\GatherAccept(i)$ invoked\atadd{ by} one correct process returned $\True$, then eventually it must return $\True$ \atreplace{for}{to} all correct processes.
\end{properties}

Thanks to the properties of AVSS and Random Secret Draw (in particular, to the {\pNotificationTotality} property), in our protocols, this assumption is trivially satisfied.
For Gather, we use the original protocol of~\cite{CanettiRabin93} (\atadd{to the best of our knowledge, it was first }described as a separate primitive in~\cite{gather-blog}).

\subsection{Bundled Approximate Agreement}
\label{sec:aa}


The last building block that we shall need is \emph{Approximate Agreement (AA)}~\cite{approx-agreement-1994}.
In a (one-dimensional) AA instance, the processes propose inputs and produce outputs (real values) so that the following properties are satisfied:
\begin{properties}
    \item[\pValidity:] the outputs of correct processes must be in the range of inputs of correct processes.

    \item[\pApproximateAgreementConsistency:] the values decided by non-faulty processes must be at most a distance $\aaEpsilon$ apart from each other.\footnote{We use $\aaEpsilon$ to distinguish it from $\epsilon$ used in the definition of Approximate Common Coin.}

    \item[\pTermination:] every non-faulty process eventually decides.
\end{properties}

In our algorithms, Approximate Agreement is always executed in \emph{bundles} of $n$ parallel instances.
For the sake of efficiency, one can treat it as a bundled problem with an input vector of size $n$, corresponding to the different instances and then, for every message, send information about all instances at the same time, but treat them separately as before. 
We call this abstraction \emph{Bundled Approximate Agreement} ($\BAA$).
BAA \textbf{should not be confused} with \emph{Multidimensional Approximate Agreement}~\cite{herlihy2013multidimensional}, which is a 
stronger abstraction than the one we rely upon.

Assuming binary inputs, the processes access BAA via the following interface:
\begin{objectinterface}
    \item[$\BAA.\RunBAA({[x_1,x_2,\dots,x_N]})$:] Launches $n$ instances of Approximate Agreement protocol, where the input for the $i$-th instance is $x_i$. %
    For a given parameter $\aaEpsilon$, the protocol is executed until $\aaEpsilon$-approximation is satisfied in every instance and then returns a vector of outputs $[y_1,y_2,\dots,y_N]$.
\end{objectinterface}

\atrev{For implementing $\BAA$, we suggest using the Approximate Agreement protocol proposed in \cite{abraham2004optimal} with resilience $f < \frac{n}{3}$.
Since in our protocols, the inputs are either $0$ or $1$, we do not need the termination detection techniques described in~\cite{abraham2004optimal} neither do we need the ``init'' phase of the protocol.
With the aforementioned \emph{BRB} and some trivial changes\footnote{\atadd{In the `report' messages, hashes of the values should be sent instead of the values themselves.}}, this implementation will give us the communication complexity of $O(n^3\lambda)$ and latency $4\cdot\log_2{(1/\aaEpsilon)}$.}


\section{Approximate Common Coin}
\label{sec:approxcoin}

\begin{algorithm}[!htb]
    \begin{smartalgorithmic}[1]
        \DeclareLocalVariable*[avssProof]{\pi_{\text{AVSS}}}
    
        \Parameters{domain $\Domain$, precision $\epsilon$}

        \algspace
        \DistributedObjects
            \State $\forall j \in \AllProc: \AVSS_j$ -- instance of Asynchronous Verifiable Secret Sharing with leader $j$
            \State $\Gather$ -- instance of Gather
            \State $\BAA$ -- instance of Bundled Approximate Agreement with precision $\aaEpsilon = \ApproxCoinEpsilon$ 
        \EndDistributedObjects
        
        \algspace
        \Function{\GatherAccept}{$j$}
            \State \Return $\True$ if received $\AVSS_j.\SharingComplete()$ \label{alg:approx:valid}
        \EndFunction

        \algspace
        \Operation{\Toss}{} \Returns{\Integer}
            \State $x = \RandomInt(\bigDomain)$ \label{line:approx:localrand}
            \State $\AVSS_i.\ShareSecret(x)$ \label{line:approx:share}
            \State $\Gather.\StartGather(\GatherAccept)$ \label{line:approx:gather}
            
            \algspace
            \State \WaitFor $\Gather.\DeliverSet(S)$ \label{line:approx:deliver}
            
             \State $\forall j \in \AllProc:$ 
                let $w_j = \begin{cases}
                    1, & j \in S,\\
                    0, & otherwise \\
                \end{cases}$ \label{line:approx:initw}
            \State $[w_1', \dots, w_n'] = \BAA.\RunBAA([w_1,\dots,w_N])$ \label{line:approx:aa}
            
           \State $\forall j \in \AllProc: \AVSS_j.\EnableRetrieve()$ \Comment{only after BAA completes} \label{line:approx:enable}
            \State $ \forall j \in \AllProc$: let $x_j = \begin{cases}
                   \AVSS_j.\Retrieve(), & if~w_j' \neq 0\\
                    0, & otherwise \\
                \end{cases}$ \label{line:approx:retrieve}

            \algspace

            \LineComment{Compute and return the final random number}
            \State \Return $\left(\left\lceil \sum_{j \in \AllProc} x_j \cdot w_j' \right\rceil  \bmod \bigDomain\right)$ \label{line:approx:output}
         \EndOperation
    \end{smartalgorithmic}

    \caption{Approximate common coin}
    \label{alg:approxcoin}
\end{algorithm}

%
%
\atrev{The main idea of this protocol is to aggregate numbers locally generated by enough different processes so that at least one of them is correct and the number it generated is truly random and uniformly distributed.
With a good aggregation function, the resulting value will also be uniformly distributed.
An example of such an aggregation function is addition modulo the size of the domain.}
Indeed, it is easy to see that, if $x$ is uniformly distributed over $[0..q{-}1]$ and $y$ is any number chosen independently of $x$, then $(x+y) \bmod q$ will also be uniformly distributed over $[0..q{-}1]$.
Another example of an aggregation operation that satisfies a similar property is bit-wise xor (as long  as the domain size is a power of 2).

However, without being able to solve consensus, we cannot just elect $f+1$ or $n-f$ processes from whom we shall take these values.
Thankfully, unlike xor, addition has one more useful property: it is continuous.
If we take two numbers, $x$ and $y$, such that $\dist{q}(x,y) \le \alpha$, 
then for any number $z$, $\dist{q}(z+x,z+y) \le \alpha$.\footnote{Recall that $\dist{q}(x,y)$ is the distance between $x$ and $y$ in the ring $\mathbb{Z}_q$, i.e., $\dist{q}(x,y)=\min\{|x-y|,q-|x-y|\}$.}
Hence, a natural idea is to \emph{approximately} elect the set of processes to provide the random inputs.

More precisely, in order to produce an \emph{approximate common coin} in the range $\DomainRange$, each process locally generates\atadd{ and secret-shares} a random number in \atreplace{the same}{this} range.
Then each process gathers a set of ids of processes\atreplace{, which}{that} have completed the sharing (line~\ref{line:approx:deliver}).

The next step is to create a binary vector with $n$ positions, where each position $j$ is set to $1$ if and only if $j$ is present in the gathered set
(line~\ref{line:approx:initw}).
This vector is then used as an input for the {\BAA} protocol (line~\ref{line:approx:aa}),
which outputs a vector $W$ of \emph{weights} such that for each position $j$, the weights received by different processes are at most $\aaEpsilon$ apart.

The value of $\aaEpsilon = \ApproxCoinEpsilon$ is chosen such that the final outputs of the coin are at most $\epsilon$ apart from each other.
For the details on how this particular \atreplace{formula}{value} was computed, see \Cref{sec:approxcoin:consistency}.

%
%
Recall that BAA ensures that the output values lie within the range of of inputs of correct processes.
Moreover, by the properties of Gather and AVSS, if at least one correct process has $j$ in its gathered set, then $j$ has correctly shared its value and it can be later retrieved by the correct processes.
Therefore if the $j$-th value is irretrievable, the $j$-th component will always be assigned weight $0$.
%
On the other hand, due to the common core property of Gather, at least $n-f$ values will have weight $1$, which, as demonstrated in \Cref{sec:approxcoin:randomness}, guarantees that the  result is uniformly distributed in the desired range.

Finally, the processes reveal the secrets (lines~\ref{line:approx:enable} and~\ref{line:approx:retrieve}) and compute the resulting random number. 

\begin{theorem} \Cref{alg:approxcoin} implements an approximate common coin.
    \begin{proof}
        \Cref{app:approxcoin} shows that the algorithm guarantees \emph{Termination}, \emph{One process randomness} and \emph{Approximate $\epsilon$-consistency}.
    \end{proof}
\end{theorem}

\myparagraph{Complexity analysis} The\atadd{ communication} complexity of our \emph{approximate common coin} can be broken down in:

\begin{enumerate}
    \item $n$\atadd{ instances of} $\AVSS.\ShareSecret$ in parallel\atadd{ $\Rightarrow$ $O(n^3 \lambda)$};
    \item One instance of $\Gather$\atadd{ $\Rightarrow$ $O(n^3\lambda)$};
    \item One instance of $\BAA$ with $\aaEpsilon=\epsilon/f$\atadd{ $\Rightarrow$ $O(n^3\lambda (\log f + \log \frac{1}{\epsilon}))$};
    \item $n$\atadd{ instances of} $\AVSS.\Retrieve$ in parallel\atadd[TODO: check]{ $\Rightarrow$ $O(n^3 \lambda)$};
\end{enumerate}

\atrev{Hence, the total communication complexity is $O(n^3 \lambda (\log f + \log \frac{1}{\epsilon}))$ with Bundled Approximate Agreement being the bottleneck.}

\atadd{The time complexity of the protocol is  $O(\log f + \log \frac{1}{\epsilon})$.}

\section{Monte Carlo Common Coin from Approximate Common Coin} \label{sec:probcoin-from-approxcoin}

\begin{algorithm}[!htb]
    \begin{smartalgorithmic}[1]
        
        \DeclareLocalDistributedObject{AC}
        
        \Parameters{domain size $\Domain$, success probability $\delta$}
        
        \algspace
        \State let $\domainFactor = \lfreplace{\lceil}{\lfloor} \frac{2}{1 - \delta} \lfreplace{\rceil}{\rfloor}$
        
        \algspace
        \DistributedObjects
            \State $\AC$ -- instance of approximate common coin with domain size $\domainFactor \Domain$ and precision $\epsilon=\frac{1}{\domainFactor \Domain}$
        \EndDistributedObjects
        
        \algspace
        \Operation{\Toss}{} \Returns{\Integer}
            \State \Return $\left\lfloor \dfrac{\AC.\Toss\atadd{()}}{\domainFactor} \right\rfloor$
        \EndOperation
    \end{smartalgorithmic}
    
    \caption{Monte Carlo Common Coin from Approximate Common Coin, code for process $i$}
    \label{alg:probcoin-from-approxcoin}
\end{algorithm}

In this section, we present a simple reduction from an approximate common coin to a Monte Carlo common coin.
The very short pseudocode is in \Cref{alg:probcoin-from-approxcoin}.

The transformation requires first to generate an approximate common coin of domain $k\Domain$ where $k$ is an integer number $\lfloor \frac{2}{1 - \delta} \rfloor$ and $\epsilon =\frac{1}{\domainFactor \Domain}$. This implies that different processes shall get values at most
$\left \lceil \frac{1}{\domainFactor \Domain}\cdot \domainFactor\Domain \right\rceil = 1$ apart.

The domain of the approximate coin is $\domainFactor$ times larger than the domain of the targeted Monte Carlo common coin. By dividing the result by \atreplace{domain factor}{$\domainFactor$}, we get the desired range of values and success probability $\delta$, where $k$ values of the \emph{approximate common coin} are mapped to one value of \emph{Monte Carlo common coin}.

\begin{theorem}
    \Cref{alg:probcoin-from-approxcoin} implements a Monte Carlo common coin with domain $\Domain$ and success probability \atrev{$\delta$}.
    \begin{proof}
        {\pTermination} and {\pUnpredictability} follow from the properties of \emph{approximate common coin}, while {\pRandomness} follows from {\pOneProcessRandomness}\atadd{ since exactly $\domainFactor$ values from the larger domain ($\domainFactor \Domain$) are mapped to each value in the smaller domain ($\Domain$)}.
        Let $x'$ be the resulting toss of the first correct process that completes $\BAA$ in its approximate common coin toss. Then, as established, every other process will be at a distance at most 1 from it. Hence, if the \atreplace{rest}{remainder} of the division of $x'$ by $\domainFactor$ is neither $0$ nor $\domainFactor-1$, every correct process decides the same value:
        
        \begin{align*}
            1-\delta &\geq \frac{2}{\domainFactor} = \frac{2}{\lfloor \frac{2}{1 - \delta} \rfloor} \geq 1-\delta
        \end{align*}
    \end{proof}
\end{theorem}

\myparagraph{Complexity analysis}

\atadd{This protocol runs a single instance of an Approximate Common Coin, with precision $\epsilon = \frac{1}{D\lfloor \frac{2}{1 - \delta} \rfloor}$.
If used with our algorithm from \Cref{sec:approxcoin}, it will take 
$O(\log f + \log \frac{1}{\epsilon}) 
= O(\log f + \log D + \log \frac{1}{1 - \delta})$ rounds of approximate agreement.}

\atadd{Hence, the overall time complexity of the protocol is $O(\log f + \log D + \log \frac{1}{1 - \delta})$ and the communication complexity is $O(n^3 \lambda (\log f + \log D + \log \frac{1}{1 - \delta}))$.}

\section{Direct Implementation of Monte Carlo Common Coin} \label{sec:probcoin}

\begin{algorithm}[!htb]
    \begin{smartalgorithmic}[1]
        \DeclareLocalVariable{candidates}
        
        \DeclareLocalVariable{tickets}
        \DeclareLocalVariable{values}
        \DeclareLocalVariable{winner}

        \Parameters{domain\atadd{ size} $\Domain$, success probability $\delta$, security parameter $\lambda$}
        \algspace
        \FunctionDeclaration
            \State $\Calibrate(w)$ -- returns the weight to apply to a ticket given that $\BAA$ returned $w$
        \EndFunctionDeclaration
        
        \algspace
        
        \DistributedObjects
            \State $\TicketDraw$ -- instance of Random Secret Draw with domain\atadd{ size} $2^\lambda$
            \State $\ValueDraw$ -- instance of Random Secret Draw with domain\atadd{ size} $\Domain$
            \State $\Gather$ -- instance of Gather\atremove{ (verifiability is not necessary for this protocol)}
            \State \atrev{$\BAA$ -- instance of Bundled Approximate Agreement with precision $\aaEpsilon$,}
            \State \atrev{\qquad\quad\, where $\aaEpsilon$ depends on the calibration function (see ``Weight calibration'' below)}

        \EndDistributedObjects

        \algspace
        \Function{\GatherAccept}{$j$} \Returns{\Boolean}
            \State \Return $\True$ iff received both $\TicketDraw.\ValueAssigned(j)$ and $\ValueDraw.\ValueAssigned(j)$
        \EndFunction
            
        \algspace
        \Operation{\Toss}{} \Returns{\Integer}
            \State $\TicketDraw.\StartRSD()$ \label{line:prob:tdraw}
            \State $\ValueDraw.\StartRSD()$ \label{line:prob:vdraw}
            \State $\Gather.\StartGather(\GatherAccept)$ \label{line:prob:gstart}
        
            \algspace
        
            \State \atrev{\WaitFor{event $\Gather.\DeliverSet(S)$}}
            \State $\forall j \in \AllProc:$ 
                let $w_j = \begin{cases}
                    1, & j \in S,\\
                    0, & otherwise \\
                \end{cases}$
            \State $[w_1', \dots, w_n'] = \BAA.\RunBAA([w_1,\dots,w_N])$ \label{line:prob:aa}
            \State $\candidates = \{ j \in \AllProc \mid w_j' > 0 \}$
            
            \algspace
            \State $\TicketDraw.\EnableRetrieve()$ 
            \State $\ValueDraw.\EnableRetrieve()$ \label{line:prob:enable-retrieve}
            \State $\tickets = \TicketDraw.\RetrieveValues(\candidates)$
            \State $\values = \ValueDraw.\RetrieveValues(\candidates)$
            
            \algspace
            \State $\winner = \argmax\limits_{j \,\in\, \candidates}  \Calibrate(w_j') \cdot \tickets[j]$ \label{line:prob:calib}
            \State \Return $\values[\winner]$ \label{line:prob:decide}
        \EndOperation
    \end{smartalgorithmic}
    
    \caption{Monte Carlo Common Coin, code for process $i$}
    \label{alg:probcoin}
\end{algorithm}

\myparagraph{Overview} The main idea of \atreplace{our \emph{Monte Carlo Common Coin}}{this} protocol is to assign to each process a random value and a random ticket. Then, using approximate agreement, the protocol is able to select the process with maximum ticket with adjustable probability of success and adopt the value corresponding to this ticket as the coin value.

\myparagraph{Tickets and values} The protocol first assigns random values and random tickets to each participant, maintaining both secret until later (lines~\ref{line:prob:tdraw} and~\ref{line:prob:vdraw}). Processes then gather a list of participants who have both drawn a ticket and a value, guaranteeing that all processes will hold sets that all intersect in at least $n-f$ participants (line~\ref{line:prob:gstart}).

\myparagraph{Approximate Agreement} In a similar manner as in the previous protocol, each process runs Bundled Approximate Agreement inputting $1$ in the dimensions corresponding to the processes it has received from 
Gather, and $0$ in the other dimensions (line~\ref{line:prob:aa}).
Similar to the previous protocol, if a process has not made a valid draw it will always be assigned weight zero, whereas if the weight is different than zero then it is possible to recover the secretly drawn number.

\myparagraph{Opening the secrets} Prior to the first decision of a correct process in $\BAA$, no secrets are leaked, as the underlying Random Secret Draw abstraction requires at least one correct process to invoke the $\EnableRetrieve$ operation before any information about the generated numbers is revealed.
After this first decision of a correct process, the secrets 
can be opened (line~\ref{line:prob:enable-retrieve}), but at this point the adversary can only induce other processes deciding values which are at most $\aaEpsilon$ apart from the first decision, which does not undermine the safety of the protocol.

\myparagraph{Decision} With the tickets and values now openly available, the processes calibrate the tickets by multiplying the plain ticket 
by a \emph{calibration function} applied to the weights.
The simplest calibration function is an identify function, the calibrated ticket of a process $i$ is simply the product of the output $i$-th output of BAA and the original ticket $t_i$. In their final steps, processes decide the value corresponding to the highest calibrated ticket (lines \ref{line:prob:calib} and \ref{line:prob:decide}).

\myparagraph{Weight calibration}

\begin{figure}
    \centering
    \input{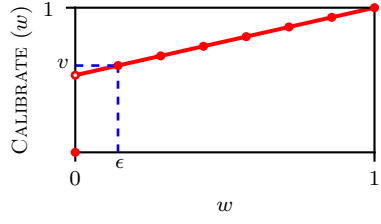}
    \vspace{-0.3cm} 
    \caption{The weight calibration function.}
    \label{fig:calibrate}
\end{figure}

As shown in \Cref{sec:probcoin-proof-without-calibration}, the protocol without weight calibration (i.e., with $\Calibrate(w) = w$) requires $\ProbcoinRoundsWithoutCalibrationWithDelta$ rounds of approximate agreement in order to achieve the success probability $\delta$.
\atrev{A similar performance is achieved by the protocol in \Cref{sec:probcoin-from-approxcoin}.}

In order to get rid of the $\log_2(n)$ part in the time complexity, we use a calibration function that is linear on $(0, 1]$ with a discontinuity point at $0$, as illustrated in \Cref{fig:calibrate}.
If $w_1 > 0$ and $w_2 > 0$ and $|w_1 - w_2| = \aaEpsilon$, then, after calibration, $|\Calibrate(w_1) - \Calibrate(w_2)| \approx \aaEpsilon \cdot (1-v)$.
This is similar in effect to running extra $\log_2 \frac{1}{1-v}$ rounds of approximate agreement, but at no extra latency cost.
We then balance the value of the parameter $v$ in such a way that, intuitively, the discontinuity at $0$ is very unlikely to cause disagreement. An example of a good value for $v$ that achieves this goal is $1 - \frac{\ln(2/(1-\delta))}{2n/3}$.
A detailed proof of the solution with weight calibration is presented in \Cref{sec:probcoin-proof-with-calibration}.
In order to achieve success probability $\delta$, $\ProbcoinRoundsWithCalibrationWithDelta$ rounds of approximate agreement are required.

\begin{theorem} \Cref{alg:probcoin} implements a Monte Carlo common coin.
    \begin{proof}
        \Cref{proof:prob} shows that the algorithm guarantees \emph{Termination}, \emph{Uniform distribution}, and \emph{Probabilistic $\delta$-consistency}.
    \end{proof}
\end{theorem}

\myparagraph{Complexity analysis}
The\atadd{ communication} complexity of our \emph{Monte Carlo common coin} can be broken down in:

\begin{enumerate}
    \item \atrev{2 instances of Random Secret Draw $\Rightarrow$ $O(n^3 \lambda)$};
    \item \atrev{1 instance of $\Gather$ $\Rightarrow$ $O(n^3 \lambda)$}
    \item One instance of $\BAA$ with $\aaEpsilon=O(1/(1-\delta))$\atadd{ $\Rightarrow$ $O(n^3 \lambda \log \frac{1}{1-\delta})$};
    \item $2n$ secret retrievals in parallel\atrev{ $\Rightarrow$ $O(n^3 \lambda)$};
\end{enumerate}

Hence, the total communication complexity is $O(n^3\lambda \log \frac{1}{1-\delta})$ with Bundled Approximate Agreement being the bottleneck.
The time complexity of the protocol is $O(\log \frac{1}{1-\delta})$.

\myparagraph{Empirical performance analysis via simulation}

\begin{figure}
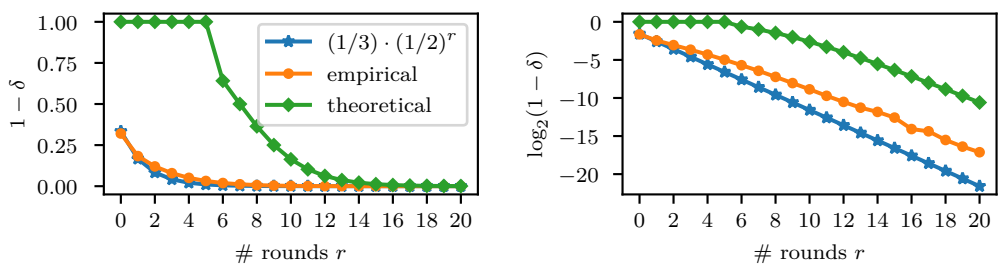

    \centering
    \input{linear_scale.pgf}
    \input{log_scale.pgf}
    \caption{\atrev{Empirical estimate of the failure probability} of the Monte Carlo common coin protocol with weight calibration\atadd{, for $n=50$,} compared to the theoretical estimate.}
    \label{fig:simulation-results}
\end{figure}


\atrev{In \Cref{sec:probcoin-proof-with-calibration}, we proved that our Monte Carlo common coin with a calibration function achieves time complexity of $O(\log \frac{1}{1 - \delta})$.
However, the derived upper bound on the exact number of rounds of approximate agreement $\left(\ProbcoinRoundsWithCalibrationWithDelta\right)$ is clearly pessimistic.
To get a better understanding of the actual performance of the protocol, we performed simulations.}

\Cref{fig:simulation-results} presents the empirical estimate of failure probability of our protocol compared to the theoretical value and the ``ideal'' value of $(1/3) \cdot (1/2)^r$, where $r$ is the number of rounds of approximate agreement.
Each point on the ``empirical'' curve is based on 3 experiments and each experiment consists of 1 million simulated executions.
In each execution, $n=50$ processes are assigned random tickets and the adversary wins if it manages to get processes in disagreement on which process has the maximum ticket.
The result of each experiment is the ratio between the number of executions where the adversary wins to the total number of executions (1 million).
Then we average over 3 such experiments to get the value for each point on the curve.

The estimate for the optimal value of the parameter $v$ in the calibration function for each number of rounds was also computed based on simulations (independent from the ones used to estimate the final failure probabilities).
Note that it may only increase the estimated empirical latency compared to the (unknown) optimal value for $v$.

\section{Applications}
\label{sec:applications}

\subsection{Binary Byzantine Agreement} \label{sec:binary-agreement}
Monte Carlo common coin can be plugged into any Byzantine Agreement protocol that makes a call to a probabilistic common coin, such as~\cite{bracha-brb,CanettiRabin93,crain2020algorithms,mostefaoui2015signature}.

Using our Monte Carlo common coin obtained via the transformation from approximate common coin (\Cref{sec:probcoin-from-approxcoin}), we get a protocol that is secure against an adaptive adversary, assumes no trusted setup or PKI, and exhibits communication complexity $O(n^3 \lambda \log n)$\atadd{ at the expense of extra $O(\log n)$ factor in time complexity}.
As far as we know, the best existing setup-free solutions that are \atadd{resilient against adaptive adversary} and tolerate up to $f<n/3$ failures exhibit\atadd{ communication complexity of} $O(n^4\lambda)$~\cite{kogias2020adkg,abraham2021reaching}. 

\subsection{Intersecting Random Subsets} \label{sec:committee-elections}

A direct application of an approximate common coin 
is a problem we call \emph{intersecting random subsets}. This problem consists of given a globally known set $S$ of cardinality $n$, to chose a subset $s \subseteq S$ of cardinality $m \leq n$.

The following variation of Gray Codes~\cite{gray-codes,coding-theory-textbook} will be instrumental:

\begin{definition} \label{def:code}
Code $C_{n,m}$ is a list of \lfreplace{$n-f$}{$\binom{n}{m}{-}1$} binary strings (called \emph{code words}) satisfying the following conditions:
\begin{itemize}
    \item for all $i$, $C_{n,m}[i]$ is a string of $m$ ones and $n-m$ zeros;
    \item every string of $m$ ones and $n-m$ zeros is present in $C_{n,m}$ exactly once;
    \item $\forall i \in \left\{1,\dots,\binom{n}{m}{-}1\right\}: C_{n,m}[i]$ and $C_{n,m}[i{-}1]$ differ in \lfremove{at most} two bits;
    \item $C_{n,m}[\binom{n}{m}{-}1]$ and $C_{n,m}[0]$ differ in \lfremove{at most} two bits.
\end{itemize}
\end{definition}

In \Cref{app:committee-elections-concrete-code}, we provide a concrete code that satisfies these properties and an algorithm (in the form of a recursive formula) that allows to efficiently generate a code word by its index.

Intuitively, this code is composed of binary strings of length $n$ and it can be read in the following manner: if the $i$-th position of the string is $1$, then $i$-th element is selected, otherwise it is considered to be left outside. 
Moreover, this code has the property that consecutive numbers differ only by swapping exactly one position set to $1$ with a position marked with a $0$. Therefore all consecutive subsets have the same fixed size and include $s-1$ common elements and differ by only one.Hence, by generating an approximate common random coin over the domain $\left\{0..\binom{n}{m}{-}1\right\}$ with parameter $\epsilon \le k \cdot \binom{n}{m}^{-1}$, processes can select subsets of size $m$ differing by at most $k$ elements.

This could be interesting for selecting a committee among $n$ users in scenarios subject to a mobile Byzantine adversary, i.e. on systems where the set of processes who display malicious behaviour changes, provided the time to corrupt a majority of processes in any given committee is higher than an asynchronous round.
%
Note that this solution provides an interesting alternative to committee elections in protocols such as Algorand~\cite{algorand}.
It not only is completely asynchronous, but it also guarantees a fixed committee size and provides a way to control the intersection of quorums obtained by different users.
Recall that in the case of Algorand, with non-zero probability, it might happen that quorums do not intersect at all.

\section{Related Work} 
\label{sec:related}

Ben-Or~\cite{BenOr83} proposed the first randomized consensus algorithm based on the ``independent choice'' common coin. 
\pkrev{
In this algorithm, every participant tosses a local random coin and with probability $2^{-n}$, the values picked by $n$ participants are identical. 
Bracha~\cite{bracha-brb} extended this algorithm to the Byzantine fault model with $f<n/3$ faulty participants.
}

\pkrev{
Rabin~\cite{Rab83} proposed an implementation of a perfect common coin based on Shamir's secret sharing~\cite{shamir1979share}, assuming trusted setup (a trusted dealer distributes  \emph{a priori} a large number of secrets).
}
Today, most protocols that allow trusted setup use the solution proposed by Cachin et al.~\cite{cachin2005random} who have described a practical perfect common coin based on threshold pseudorandom functions (tPRF), assuming that a trusted dealer distributes a short tPRF key.
These protocols use the pre-distributed randomness in a clever way to obtain multiple numbers that are computationally indistinguishable from random (much like a PRNG~\cite{prng}).
In contrast, our algorithms do not assume trusted setup.

\pkrev{In the setup-free context}, Canetti and Rabin~\cite{CanettiRabin93} proposed a weak common coin algorithm based on asynchronous verifiable secret sharing (AVSS), which resulted in an efficient randomized \emph{binary} consensus.
In designing our protocols, we make use of multiple ideas from this work, taking into account recent improvements, such as the use of Aggregatable Publicly Verifiable Secret Sharing~\cite{gurkan2021aggregatable}, suggested in \pkrev{a similar context} by 
Abraham et al.~\cite{abraham2021reaching}. \lfadd{We give a modern version of their coin using our building blocks in \Cref{app:cr93}.} 
%


Using standard PKI cryptography, Cohen et al.~\cite{cohen2020coincidence} built two common coins relying on VRFs. 
The first one with resilience $f<(1/3-\epsilon)$ achieves \lfadd{success rate} $\delta = \frac{18\epsilon^2+24\epsilon-1}{6(1+6\epsilon)}$.
The second coin involves sampling committees among the processes and also guarantees a constant success rate that depends on the system's \pkrev{resilience}.

Kogias et al.~\cite{kogias2020adkg} \pkrev{proposed a relaxed abstraction} called \emph{eventually perfect common coin}. 
They first build a weak distributed key generator (wDKG) which is a protocol that never terminates: each party outputs a sequence of candidate keys to be used \lfadd{for encryption and decryption} with the property that they will eventually agree on a set of keys.
\pkrev{This mechanism can replace the trusted setup in~\cite{cachin2005random}.
Moreover, it can be shown that the participants may disagree on the set of keys at most $f+1$ times.} 
The resulting coin eventually terminates with a perfect result, hence its name.
\pkadd{In contrast, the challenge of our work was to devise (one-shot) unbiased common coins with provable termination}. 

Gao et al.~\cite{gao2021efficient} combined VRFs and AVSS to produce the first random coin which has $O(n^3\lambda)$ \lfadd{communication} complexity.
\pkrev{With the advent of new broadcast and APVSS implementations, \atadd{the classic protocol of }\cite{CanettiRabin93} gains the same complexity (see \Cref{app:cr93}).}
They created a weak form of \emph{Gather} called \emph{core-set selection} in which $f+1$ correct participants share at least $n-f$ VRFs coming from different processes. 
They then use AVSS to build the seeds for VRFs and whenever the highest VRF is in the common core, the nodes successfully agree in the coin outcome. 
%
Their protocol assumes the static adversary, \atrev{but it can be made adaptive with a \emph{relatively weak} form of trusted setup: a single common random number must be published after the public keys of the participants are fixed}.
\atadd{In contrast, our protocols do not assume any form of trusted setup.}
\pkrev{Assuming\atremove{ the} static adversary, our protocols can achieve the same communication cost while additionally enabling parameterized success rate.}

\pkrev{Our Monte Carlo coin\atadd{ from \Cref{sec:probcoin}} was inspired by the Proposal Election protocol recently introduced by Abraham et al.~\cite{abraham2021reaching}.
Technically, it is not a common coin \emph{per se} but\atadd{ it} uses elements of it.} 
In this protocol, every party inputs some externally valid value, and with \pkadd{a} constant probability, all parties output the same value that was proposed by a non-faulty party chosen at random.
\atadd{Intuitively, the main contribution of the protocol in \Cref{sec:probcoin} is the use of approximate agreement to \emph{amplify} the success probability.}

In the \emph{full information} model, without using cryptography, King and Saia~\cite{king2016byzantine} \pkrev{observed} that the strength of the Byzantine adversary is in its anonymity, \pkrev{but it cannot bias the coin indefinitely without being detected.}
Even though their Byzantine agreement algorithm with polynomial expected time does allow the adversary to bias the coin, 
but amended this with statistical tests aiming at detecting this kind of deviation and evicting misbehaving participants.
\pkadd{The resilience level of this algorithm is, however, orders of magnitude lower than $n$ ($f<n/400$ in the best case).}
\pkrev{
Huang et al.~\cite{huang2022byzantine} recently extended this work to achieve the resilience of $f<n/4$.}
\pkadd{
In contrast, our algorithms use cryptographic tools to produce unbiased (approximate) outputs and maintain optimal resilience $f<n/3$.  
}

\lfadd{Monte Carlo protocols are randomized algorithms that have a fixed number of rounds and yield results that are correct with a given probability, while Las Vegas protocols always give the correct results but do not have a fixed number of rounds. Notice that Las Vegas algorithms must have a fixed probability of terminating every round and thus can be converted into Monte Carlo by stopping after a fixed number of rounds and deciding a random value if termination is not attained.}

\lfadd{Such a transformation could be applied to \cite{kogias2020adkg}, but since their latency is a function of $O(f)$, the resulting Monte Carlo common coin would also have a latency which is a function of $f$, while our solution is independent of this parameter. Another option would be to create a set of keys using \cite{abraham2021reaching} which could be run a fixed number of rounds and then to use \cite{cachin2005random}. Since their expected number of rounds is not a function of $f$, this transformation would have the same asymptotic complexity as ours, but it would include many unnecessary message delays from the consensus protocol, from the verifiable gather and other parts of their ADKG that are not present in a direct implementation such as ours.}


\section{Conclusion}
\label{sec:conclusion}

\atrev{%
In this paper, we suggest 2 new relations of the common coin primitive implementable in a fully asynchronous environment.
We provide efficient implementations based on a range of novel techniques.
Our protocols are the first use of approximate agreement to generate random numbers, which is used to keep decided values close in the approximate common coin protocol and to increase the probability of agreement in the direct implementation of the Monte Carlo common coin.}
Moreover, we also introduced elements of coding theory that were not previously applied to the distributed computing realm in the solution of what we called intersecting random subsets.

Further studies are necessary to explore the full potential of using these new abstractions in the design of distributed protocols and to understand the theoretical limits of their performance.
%

\myparagraph{Tight performance analysis \atreplace{of}{for} probably common coin}

In \Cref{sec:probcoin-proof-without-calibration,sec:probcoin-proof-with-calibration}, we proved that in order to achieve success probability $\delta$, our probably common coin protocol requires $\ProbcoinRoundsWithoutCalibrationWithDelta$ or $\ProbcoinRoundsWithCalibrationWithDelta$ rounds of approximate agreement, depending on whether weight calibration is used.
While\atremove{, according to our simulations,} it seems to correctly reflect the asymptotic behaviour of the actual distribution, these bounds seem to be rather pessimistic when only a few rounds of approximate agreement are run.

For example, according to these bounds, with weight calibration, we will need $8$ rounds in order to achieve $\delta = \frac{2}{3}$.
However, in practice, $\delta = \frac{2}{3}$ is achieved with 0 rounds of approximate agreement as with probability at least $2/3$ a process in the common core will have the largest ticket.
\atadd{As we estimated empirically through simulations, the actual $\delta$ that we obtain after running $8$ rounds of approximate agreement is around $0.993$ instead of $2/3$.}

\myparagraph{Non-linear calibration function for Monte Carlo common coin.}

Weight calibration is necessary to achieve the latency of \atreplace{$O(\log(1/\epsilon))$}{$O(\log \frac{1}{1-\delta})$} rounds in our Monte Carlo common coin protocol.
We chose a concrete linear function because it was relatively simple to analyze (as we could do a reduction to the case without the calibration).
However, this function is unlikely to be optimal. 
The extra \atreplace{$\log_2(\log_2(1/Q))$}{$\log_2(\log_2 \frac{1}{1-\delta})$} part in the resulting estimate on the number of rounds of approximate agreement is likely to be due to sub-optimal choice of the weight calibration function.



\myparagraph{Approximate common coin without extra $\log_2(f)$ rounds}

Using the magic of weight calibration, for Monte Carlo common coin, we managed to achieve $O(\log(1/\epsilon))$ time complexity, which is likely to be optimal.
However, our approximate common coin protocol requires $\ApproxCoinRounds$ rounds of approximate agreement and, hence, its time complexity depends on two variables: $f$ and $\epsilon$.
In some applications, $\epsilon$ may be constant and $\log_2(f)$ can become the bottleneck.

Creating a protocol without this extra delay or proving a $\Omega(\log(f))$ lower bound would be an interesting contribution to the understanding of the approximate common coin abstraction. It would also mean that the transformation from approximate common to Monte Carlo common would result in more efficient coins of the latter type.

\bibliographystyle{plainurl}
\bibliography{main}

\newpage
\appendix

\section{Impossibility of an Asynchronous Perfect Common Coin} \label{app:impossibility}

A perfect common coin protocol makes sure that the correct processes agree on a random number taken uniformly over a specific domain.  

By reusing the classical arguments of the impossibility of asynchronous consensus~\cite{FLP85}, we are going to show that no asynchronous perfect common coin protocol may exist, if a single process is allowed to fail by crashing.   
Recall that we consider protocols in which processes may non-deterministically choose its actions based on local coin tosses.

Formally, a protocol provides each process with an automaton that, given the process' state and an input (a received message and a result of a local random coin toss), 
produces an output (a finite set of messages to send and/or the application output). 
We assume that the automaton itself is deterministic, i.e., all non-determinism is delegated to the outcomes of local random coins. 
A step of process $p$ is therefore a tuple $(p,m,r)$,
where $m$ is the message $p$ receives (can be $\bot$ if no message is received in this step) and $r$ is the outcome of its local random coin.

%

A \emph{configuration} of the protocol assigns a local state to each process' automaton and a set of messages in transit (we call it the \emph{message buffer}). 
The \emph{initial configuration} $C_{\textit{init}}$ assigns the initial local state to each process and assumes that there are no messages in transit.
A step $s=(p,m,r)$ is \emph{applicable to a configuration} $C$ if $m$ is $\bot$ or $m$ is in the message buffer of $C$.   
The result of $s$ applied to $C$ is a new configuration $C.s$, where, based on the automaton of $p$, the local state of $p$ is modified and finitely many messages are added to the message buffer.

%
%
A sequence of steps $E=s_1,s_2,\ldots$ is applicable to a configuration $C_0$ if each $s_i$, $i=1,2,\ldots$ is applicable to $C_{i-1}$, where, for all $i\geq 1$, $C_i=C_{i-1}.s_i$.
We use $C_0.E$ to denote the resulting configuration;
we also say that $C=C_0.E$ is an \emph{extension of $C_0$}.
By convention, $C$ is a trivial extension of $C$.

An \emph{execution} of the protocol is a sequence of steps $E$ applicable to $C_{\textit{init}}$.
A finite execution $E$ results in a \emph{reachable} configuration $C_E$.
Somewhat redundantly, we call a sequence of steps applicable to a reachable configuration $C_E$ an \emph{extension of $E$}. 
It immediately follows that steps of disjoint sets of processes commute:
\begin{lemma}
\label{lem:local}
Let $C$ be a reachable configuration, and $E$ and $E'$ be sequences of steps of disjoint sets of processes. If both $E$ and $E'$ are applicable to $C$, then $C.E.E'$ and $C.E'.E$ are identical reachable configurations.  
\end{lemma}

In an infinite execution, a process is \emph{correct} if it executes infinitely many steps. 
We assume that every message $m$ that is sent to a correct process $p$ is eventually received, i.e., the execution will eventually contain a step $(p,m,-)$.
As we assume that at most one process is allowed to fail by crashing, we only consider infinite executions in which at least $n-1$ out of $n$ processes are correct.

Without loss of generality, we assume that the implemented coin is binary: the correct processes either all output $0$ or all output $1$.    

A configuration $C$ is called \emph{bivalent} if it has an extension $C.E_0$ in which some process outputs $0$ and an extension $C.E_1$ in which some process outputs $1$.
Notice that any configuration preceding a bivalent configuration must be bivalent. 
Also, no process can produce a random-coin output in a bivalent configuration: otherwise, we get an execution in which two processes disagree on the output. 

A protocol configuration that is not bivalent is called \emph{univalent}: $0$-valent if it has no extension in which $1$ is decided or $1$-valent otherwise.  

\begin{lemma}
\label{lem:bivalent}
The initial configuration $C_\textit{init}$ is bivalent. 
\end{lemma}
\begin{proof}
The algorithm must output each of the two values with a positive probability. 
Thus, for each $v\in\{0,1\}$, there exists an assignment of local coin outcomes and a message schedule that result in an execution with outcome $v$.
\end{proof}

\begin{lemma}
\label{lem:solo}
Let $C$ be a reachable configuration, and $E$ and $E'$ be sequences of steps of a process $p$ applicable to $C$. 
If $C.E$ and $C.E'$ are univalent, then they have the same valencies.
\end{lemma}
\begin{proof}
The difference between $C.E$ and $C.E'$ consists in the local state of $p$ and the message buffer. 
Since the algorithm is required to tolerate a single crash fault, there must exist a sequence of steps $E''$ that does not include any steps of $p$ such that some process $q$ outputs a value $v \in \{0,1\}$ in $C.E''$.
Moreover, as $E''$ contains no steps of $p$, $E''$ is also applicable to both $C.E$ and $C.E'$. 
But as the two configurations have opposite valences, and $q$ decides the same value in $C.E.E''$ and $C.E'.E''$.
Thus, if $C.E$ and $C.E'$ are univalent, then they must have the same valencies.
\end{proof}
Following the steps of ~\cite{FLP85}, we show that the protocol must have a \emph{critical} configuration $D$ for which:
\begin{itemize}
    \item there exist steps $s_p$ and $s_q$ of processes $p$ and $q$ such that both $s_p$ and $s_q.s_p$ are applicable to $D$
    
    \item $D.s_p$ is $0$-valent; 
    \item $D.s_q.s_p$ is $1$-valent.
\end{itemize}

\begin{lemma}
\label{lem:critical}
There must exist a critical configuration.
\end{lemma}
\begin{proof}
Starting with the initial (bivalent by Lemma~\ref{lem:bivalent}) configuration $C:=C_\textit{init}$, we pick process $p:=p_1$ and check if there exists $C'$, an extension of $C$, and $s_p$, a step of $p$ applicable to $C'$ in which the oldest message to $p$ in the message buffer is consumed (if any), such that $C'.s_p$ is bivalent. 
If this is the case, we set $C$ to $C'.s_p$, pick up the next process $p:=p_2$.
Again, if there exists $C'$, an extension of $C$, and $s_p$, a step of $p$ in which the oldest message to $p$ in the message buffer is consumed, such that $C'.s_p$ is bivalent, then we set $C$ to $C'.s_p$.
We repeat the procedure, each time picking the next process (in the round-robin order, i.e., after $p_n$ we go to $p_1$, etc.), as long as it is possible.

The first observation is that the procedure cannot be repeated indefinitely.
Indeed, otherwise, we obtain an infinite sequence of steps in which every process takes infinitely many steps and receives every message sent to it (and which is, thus, an execution of the protocol) that goes through bivalent configurations only.
Hence, this execution cannot produce an output---a contradiction with the Termination property.

Thus, there must exist a bivalent configuration $C$ and a process $p$, such that for each $C'$, an extension of $C$ and each $s_p=(p,m,-)$, a step of $p$ applicable to $C'$, $C'.s_p$ is univalent, where $m$ is the oldest message addressed to $p$ in $C$.  

Let $s_p=(p,m,r)$ be a step of $p$ applicable to $C$.
Without loss of generality, let $C.s_p$ be $0$-valent.

As $C$ is bivalent, it must have an extension $E=e_1,\dots,e_k$ such that $C.E$ is $1$-valent. 

Let $\ell$ be the largest index in $\{1,\dots,k\}$ such that either $s_p$ is not applicable to $C_{
\ell}=C.e_1,\ldots,e_{\ell}$ or $C_{\ell}.s_p$ is $1$-valent.
Such an index exists, as $C.s_p$ is $0$-valent $C.E$ is $1$-valent and for any $C'$, an extension of $C$, if $s_p$ is applicable to $C'$, then $C'.s_p$ is bivalent. 

Suppose first that $s_p$ is not applicable to $C_{\ell}$. 
As $s_p$ is applicable to every configuration $C_{j}=C.e_1,\dots,e_{j}$, $j=1,\ldots,\ell-1$, $e_{\ell}$ must be be of the form $(p,m,r')$, i.e., $e_{\ell}$ must consume the message received in $s_p=(p,m,r)$.
    By our assumption $C_{\ell}=C_{\ell-1}.(p,m,r')$ is univalent. 
    As $C_{\ell}$ has a $1$-valent descendant $C'$, it must be $1$-valent too.
    
    But, by Lemma~\ref{lem:solo}, $C_{\ell-1}.s_p=C_{\ell-1}.(p,m,r)$ and $C_{\ell}=C_{\ell-1}.(p,m,r')$, must have the same valencies---a contradiction.
    
Thus, $C_{\ell}.s_p$ is $1$-valent.
Hence, we have a bivalent configuration $D=C_{\ell-1}$ and steps $s_p$ and $s_q=e_{\ell}$ such that both $s_p$ and $s_q.s_p$ are applicable to $D$.
Moreover, $D.s_p$ is $0$-valent and $D.s_q.s_p$ is $1$-valent. 
Thus, we get a critical configuration.
\end{proof}
Finally, we establish a contradiction by showing that:

\begin{lemma}
\label{lem:critical-imp}
No critical configuration may exist.
\end{lemma}
\begin{proof}
By contradiction, let a critical configuration $D$ exist, and let $s_p$ and $s_q$ be steps of $p$ and $q$ applicable to $D$ such that $D.s_p$ is $0$-valent and $D.s_q.s_p$ be $1$-valent. 

If $s_q$ is a step of $p$, then Lemma~\ref{lem:solo} establishes a contradiction.  

Otherwise, consider any infinite execution going through to $D.s_p$ in which all processes but $q$ take infinitely many steps in this execution after $D.s_p$.
By the Termination property, there must exist a finite sequence of steps $E$ such that some process outputs a value $v\in\{0,1\}$ in $D.s_p.E$. 
As $D.s_p$ is $0$-valent, $v=0$.  
Moreover, as $E$ contains no steps of $q$ and $p$ has the same state in $D.s_p$ and $D.s_q.s_p$, $E$ is also applicable to $D.s_q.s_p$. 
Thus, $0$ is also decided in $D.s_q.s_p.E$---a contradiction with the assumption that $D.s_q.s_p$ is $1$-valent.
\end{proof}
Lemmata~\ref{lem:critical} and~\ref{lem:critical-imp} imply:

\begin{theorem}
\label{th:imp}
 There does not exist a $1$-resilient asynchronous random coin protocol.   
\end{theorem}

\section{Proof of Correctness of the Approximate Common Coin Algorithm}
\label{app:approxcoin}

\subsection{Termination}
Consider a correct process $i$.
As every correct process $j$ first performs $\AVSS_j$, $\GatherAccept(j)$ will eventually evaluate to true at $i$.
Moreover, $\GatherAccept$ trivially satisfies the {\pAcceptTotality} property of from \Cref{sec:gather}.
Thus, $i$ will eventually witness $\Gather.\DeliverSet(S)$ (line~\ref{line:approx:deliver}).
Termination of Bundled Approximate Agreement (BAA) ensures that eventually $i$ receives a weight vector.
Recall that, as the output of BAA at any position $j$ lies within the inputs of correct processes, it may carry a positive value only if $j$ indeed completed sharing its value in $\AVSS_j$. 
At least $n-f$ correct processes eventually invoke $\AVSS_j.\EnableRetrieve()$ for each process (line~\ref{line:approx:enable}).
Thus, every such input will be successfully retrieved (line~\ref{line:approx:retrieve}) and the coin output will be produced (line~\ref{line:approx:output}). 

\subsection{\pOneProcessRandomness}
\label{sec:approxcoin:randomness}
Let $i$ be the first correct process to return from the BAA invocation (line~\ref{line:approx:aa}).
We are going to show that the value $i$ is uniformly distributed over $[0,\ldots,\bigDomain-1]$.   

By the algorithm, the value is coming from a scalar product of the weight vector and the vector of retrieved secrets.
%
The {\pBindingCommonCore} property of the Gather abstraction guarantees that there is a set $\CommonCore$ of at least $n-f$ valid proofs delivered at all correct processes.   
Thus, the properties of BAA, 
there are at least $n-f$ indices in the weight vector $W$ set to $1$, and at least $f+1$ of these indices belong to correct processes.
Thus, when process $i$ retrieves the values from AVSS (line~\ref{line:approx:retrieve}), at least $f+1$ of them come from correct processes.
Each of these values were chosen uniformly at random in the range $[0,\dots,\bigDomain{-}1]$ (line~\ref{line:approx:localrand}).

Therefore, the result is computed from the sum of at least $n-f$ retrieved values (modulo $\bigDomain$), 
and at least $f+1$ of these values come from the correct processes.
The retrieved values coming from the Byzantine processes must have been previously shared in AVSS and, thus, must have been chosen before any correct process reveals its input.
Thus, regardless of the values shared by the adversary, the result computed in line~\ref{line:approx:output} is uniformly distributed in the range $\DomainRange$.

\subsection{\pApproximateEpsConsistency} 
\label{sec:approxcoin:consistency}


Let $w_i^j$ denote the weight of process $i$ received by process $j$ from the approximate agreement protocol on line~\ref{line:approx:aa}.
The output of the coin for process $k$, which we will denote as $R_k$, is computed as follows (line~\ref{line:approx:output}): $\left \lceil \sum_{i \in [n]} x_i \cdot w_i^k \right \rceil \bmod D$.

Recall that $\dist{q}(x,y) = \min\left\{ |x - y|, q - |x - y| \right\}$. 
We need to show that, for any two processes $k$ and $l$, $\dist{D}(R_k, R_l) \le \lceil \epsilon D \rceil$, where $\epsilon$ is the parameter of the approximate common coin.
For this, we will need the following three inequalities that can be easily verified:

\begin{gather} 
    \forall q \in \{2,3,4,\dots\};~a,b \in \{1,2,3,4,\dots\}: \dist{q}(a \bmod q, b \bmod q) \le |a - b| \label{ineq:dist-to-abs}\\
    \forall a,b \in \mathbb{R}_{\ge 0}: \lceil a \rceil - \lceil b \rceil \leq \lceil a - b \rceil \label{ineq:ceil}\\
    \forall a \in \mathbb{R}_{\ge 0}: |\lceil a \rceil| \leq \lceil |a| \rceil \label{ineq:real}
\end{gather}


\atrev{Let $\CommonCore$ denote the set of processes in the common core of the Gather protocol.
Due to the {\pValidity} property of Bundled Approximate Agreement and the way inputs to BAA are formed (line~\ref{line:approx:initw}), $\forall i \in \CommonCore: w_i^k = w_i^l = 1$.}


%

%

Let us now derive a bound on $d(R_k, R_l)$.
\begin{align*}
    \dist{D}(R_k, R_l) 
        & = \dist{D}\left( \left \lceil \sum_{i \in [n]} x_i \cdot w_i^k \right \rceil \bmod D, \left \lceil \sum_{i \in [n]} x_i \cdot w_i^l \right \rceil \bmod D \right) & \text{By definition.}\\
        & \leq \left| \left \lceil \sum_{i \in [n]} x_i \cdot w_i^k \right \rceil - \left \lceil \sum_{i \in [n]} x_i \cdot w_i^l \right \rceil \right| & \text{Applying (\ref{ineq:dist-to-abs})} \\
        & \leq \left| \left \lceil \sum_{i \in [n]} x_i \cdot w_i^k - \sum_{i \in [n]} x_i \cdot w_i^l \right \rceil \right| & \text{Applying (\ref{ineq:ceil})}\\
        & = \left \lceil\left| \sum_{i \in [n]} (w_i^k-w_i^l) x_i  \right|\right \rceil & \text{\atadd{Applying (\ref{ineq:real})}}\\
        & = \left \lceil\left| \sum_{i \not\in S^*} (w_i^k-w_i^l) x_i  \right|\right \rceil & \text{Since } \forall i \in S^*: 
        \atrev{|w_i^k - w_i^l|} = 0\\
        & \leq \left \lceil \sum_{i \not\in S^*} |w_i^k-w_i^l|D  \right \rceil & \text{Since } \forall i: 
        \atrev{0 \le x_i} < D\\
        & \leq \left \lceil \sum_{i \not\in S^*} \aaEpsilon D  \right \rceil &\text{By properties of BAA outputs}\\
        & \le \left\lceil f\aaEpsilon\bigDomain \right\rceil\\
\end{align*}

For any $\epsilon \in (0, 1]$, if we want the output values to be at a distance of at most $\lceil \epsilon D\rceil$ apart, we can set $\aaEpsilon = \ApproxCoinEpsilon$.
Hence, in order to reach approximate $\epsilon$-consistency, we need to run the protocol for $
    \log_2 \dfrac{1}{\aaEpsilon} 
    =  \log_2 f + \log_2(1 / \epsilon)
$
rounds.

\section{Proof of Correctness of the Monte Carlo Common Coin Algorithm}
\label{proof:prob}

\DeclareSectionVariable*[PFailure]{P[\text{failure}]}
\DeclareSectionVariable*[PSuccess]{P[\text{success}]}
\DeclareSectionCommand{\first}{^{\mathrm{1st}}}

\subsection{Termination}
$\BAA$ provides termination and also the guarantee that if all non-faulty processes input $0$ in a given dimension, then every process decides $0$. For this reason, if a value is irretrievable, then the weight given to it by correct processes is zero and therefore they will not be blocked trying to retrieve them. 

\subsection{Uniform distribution}
Any decided value must be a value random secretly drawn and any such values are composed from $n-f$ values locally generated by different processes. Therefore, at least one of these values will be truly uniformly random, and whenever a uniform value over a domain $D$ is added to any other value in mod $D$, the result will always be a uniformly random over $D$. 

\subsection{Proof of probabilistic $\delta$-consistency without weight calibration}
\label{sec:probcoin-proof-without-calibration}

Let $T_i$ denote the ticket of process $i$.
We assume that $\forall i: T_i \sim U(0,1)$.
Note that, in practice, the tickets will have to be binary encoded as fixed precision numbers.
However, if the number of bits in the encoding of a ticket is proportional to the security parameter $\lambda$, then the estimate that we give with the assumption that the tickets are real numbers holds true with all but negligible probability.

Let $w_i^j$ denote the weight of process $i$ received by process $j$ from the approximate agreement protocol, $w_i^j \in [0,1]$,
and let $[w_1\first, \dots, w_n\first]$ denote the weights received by the correct process that was the first to complete approximate agreement.

Let $X_i^j = w_i^j \cdot T_i$ and $X_i\first = w_i\first \cdot T_i$.
Due to the properties of RSD, since no correct process invokes $\RSD.\EnableRetrieve$ until it completes approximate agreement, the tickets are unpredictable to the adversary until $[w_1\first, \dots, w_n\first]$ are fixed.
Hence, $X_1\first, \dots, X_n\first$ are independent random variables, and $X_i\first \sim U(0, w_i\first)$.

Let $w_i^M = \max\limits_{j} \{ w_i^j \}$, $w_i^m = \min\limits_{j} \{ w_i^j \}$,
$X_i^M = w_i^M \cdot T_i$, and $X_i^m = w_i^m \cdot T_i$.
Unlike the first received weights, the maximum and the minimum weights can be manipulated by the adversary already after it learns the tickets.
Hence, $X_1^M, \dots, X_n^M$, as well as $X_1^m, \dots, X_n^m$, are not independent and we cannot know their distributions.
However, by the properties of approximate agreement, we know that 
$w_i^M \le \min\{w_i\first + \epsilon, 1\}$ and $w_i^m \ge \max\{w_i\first - \epsilon, 0\}$.

Let $i_M = \arg\max\limits_{i} X_i\first$.
With a slight abuse of notation, we use subscript $M$ instead of $i_M$.
I.e., $T_M = T_{i_M}$, $X_M\first = X_{i_M}\first$, $X_M^M = X_{i_M}^M$, $X_M^m = X_{i_M}^m$, and so on.

\begin{theorem}
    Without weight calibration, running the approximate agreement for $\ProbcoinRoundsWithoutCalibrationWithQ$ rounds is sufficient to guarantee that the probability of not having consistency is at most $Q$.
\end{theorem}
\begin{proof}
    A sufficient condition for all processes to output the same value would be if $\forall i \neq i_M: X_i^M < \lfreplace{X_i^m}{X_M^m}$.
    Indeed, it would mean that, even if process $j$ gets the minimum possible weight for process $i_M$ and maximum possible weights for all other processes $i \neq i_M$, $j$ would still select the value associated with $i_M$.
    Hence, we provide an upper bound on the probability of failure, $\PFailure$, by estimating the probability there exists $i$ such that $X_i^M \ge X_M^m$.
    As we demonstrate bellow, in \Cref{eq:probably-common-coin-main}, $\PFailure \le 8n\epsilon$.
    Hence, we can conclude that $\epsilon \le \frac{Q}{8n}$ is sufficient to guarantee that $\PFailure \le Q$.
    Since the approximate agreement protocol we use requires $\lceil\log_2 \frac{1}{\epsilon}\rceil$ rounds,
    running it for $\lceil \log_2 \frac{8n}{Q} \rceil \le \ProbcoinRoundsWithoutCalibrationWithQ$ rounds is sufficient to achieve $\epsilon = \frac{Q}{8n}$ and, hence, $\PFailure \le Q$.
\end{proof}

\begin{floatingequation}[H]
    \begin{equation*}
    \begin{aligned}
        \PFailure 
        &\le P[\exists i \neq i_M : X_i^M \ge X_M^m] \\
        &\le P[\exists i \neq i_M : X_i\first + \epsilon T_i \ge X_M\first - \epsilon T_M] \\
        &\le P[\exists i \neq i_M : X_i\first \ge X_M\first - 2\epsilon] \\
        &=   \sum\limits_{k=1}^n P[i_M = k] \cdot P[\exists i \neq k : X_i\first \ge X_k\first - 2\epsilon \mid i_M = k] \\
        &=   \sum\limits_{k=1}^n \int\limits_0^{w_k\first} p_{X_k\first}(x) \cdot P[i_M = k \mid X_k\first = x] \cdot P[\exists i \neq k : X_i\first \ge X_k\first - 2\epsilon \mid i_M = k \land X_k\first = x] \dif x \\
        &=   \sum\limits_{k=1}^n \int\limits_0^{w_k\first} p_{X_k\first}(x) \cdot \left (\prod\limits_{i \neq k} P[X_i\first \le x] \right) \cdot P[\exists i \neq k : X_i\first \ge x - 2\epsilon \mid X_i\first < x] \dif x \\
        & \text{(applying the union bound)} \\
        &\le \sum\limits_{k=1}^n \int\limits_0^{w_k\first} p_{X_k\first}(x) \cdot \left (\prod\limits_{i \neq k} P[X_i\first \le x] \right) \cdot \left( \sum\limits_{i \neq k} P[X_i\first \ge x - 2\epsilon \mid X_i\first \le x] \right) \cdot \dif x \\
        &=   \sum\limits_{k=1}^n \int\limits_0^{w_k\first} p_{X_k\first}(x) \cdot \left (\prod\limits_{i \neq k} P[X_i\first \le x] \right) \cdot \left( \sum\limits_{i \neq k} \frac{P[X_i\first \ge x - 2\epsilon \land X_i\first \le x]}{P[X_i\first \le x]} \right) \cdot \dif x \\
        &=   \sum\limits_{k=1}^n \int\limits_0^{w_k\first} p_{X_k\first}(x) \cdot \sum\limits_{i \neq k} \left( P[X_i\first \ge x - 2\epsilon \land X_i\first \le x] \cdot \prod\limits_{j\notin \{k,i\}} P[X_j\first \le x] \right) \cdot \dif x \\
        &\le \sum\limits_{k=1}^n \int\limits_0^{w_k\first} p_{X_k\first}(x) \cdot \sum\limits_{i \neq k} \left( \frac{2\epsilon}{x} \cdot \prod\limits_{j\notin \{k,i\}} P[X_j\first \le x] \right) \cdot \dif x \\
        & (\text{since } \forall j \in S: P[X_j\first \le x] = P[T_j\first \le x] = x \text{ and } |S \setminus \{k, i\}| \ge n-f-2) \\
        &\le \sum\limits_{k=1}^n \int\limits_0^{w_k\first} p_{X_k\first}(x) \cdot \sum\limits_{i \neq k} \left( \frac{2\epsilon}{x} \cdot x^{n-f-2} \right) \cdot \dif x
         =  \sum\limits_{k=1}^n \int\limits_0^{w_k\first} p_{X_k\first}(x) \cdot \sum\limits_{i \neq k} \left( 2\epsilon \cdot x^{n-f-3} \right) \cdot \dif x \\
        &\le  \sum\limits_{k=1}^n \int\limits_0^{w_k\first} p_{X_k\first}(x) \cdot 2n\epsilon \cdot x^{n-f-3} \cdot \dif x
         =  \sum\limits_{k=1}^n \frac{2n\epsilon}{w_k\first} \int\limits_0^{w_k\first} x^{n-f-3} \dif x \\
        &=  \sum\limits_{k=1}^n \frac{2n\epsilon}{w_k\first} \cdot \frac{(w_k\first)^{n-f-2}}{n-f-2}
         =  \sum\limits_{k=1}^n \frac{2n\epsilon}{n-f-2} \cdot \frac{(w_k\first)^{n-f-2}}{w_k\first} \\
        & \text{(since $n \ge 3f+1$ and $f \ge 1$, it can be easily verified that $n-f-2 \ge n/4 \ge 1$)} \\
        &\le \sum\limits_{k=1}^n \frac{2n\epsilon}{n/4} \cdot \frac{w_k\first}{w_f\first} = \sum\limits_{k=1}^n 8\epsilon = 8n\epsilon \\
        %
    \end{aligned}
    \end{equation*}
    
    \caption{Estimate on the failure probability of Monte Carlo common coin.}
    \label{eq:probably-common-coin-main}
\end{floatingequation}


\subsection{Proof of probabilistic $\delta$-consistency with weight calibration}
\label{sec:probcoin-proof-with-calibration}

\DeclareSectionVariable*[todoC]{\textcolor{red}{todoC}}
\DeclareSectionVariable*[newEpsilon]{\ensuremath{\epsilon'}}

Let $Q = 1 - \delta$ be the target error probability.
Let $r$ be the number of rounds of approximate agreement.
Our goal in this section is to show that, when weight calibration is applied as described in \Cref{sec:probcoin}, $r = \ProbcoinRoundsWithCalibrationWithQ$ is sufficient
to make sure that the probability of not reaching agreement is below $Q$.
$\epsilon = 2^{-r}$ is the maximum difference between the outputs of Approximate Agreement for two correct processes in the same coordinate.

Recall that we apply this optimization only when $n > \ProbcoinCalibrationMinimumNWithQ$.
Otherwise, we just use the protocol without weight calibration.
According to the proof from \Cref{sec:probcoin-proof-without-calibration},\atadd{ in this case,} we shall need 
$\ProbcoinRoundsWithoutCalibrationWithQ
\le 3 + \lceil \log_2\left(\ProbcoinCalibrationMinimumNWithQ\right) + \log_2(1/Q)\rceil
\le \ProbcoinRoundsWithCalibrationWithQ$ rounds.

\smallskip
Let $v = 1 - \frac{\ln(2/Q)}{2N/3}$
(the choice of $v$ is dictated by the proof of \Cref{lem:probcoin:calibration-failure}).

Let the weight calibration function $\Calibrate(w)$ be 
$\begin{cases}
    0, & \text{if } w = 0 \\
    \frac{w-\epsilon}{1-\epsilon} + \frac{1-w}{1-\epsilon}\cdot v, & \text{otherwise}
\end{cases}$\\

In other words, $\Calibrate(0) = 0$ and $\forall w \in (0, 1]: \Calibrate(w)$ is a linear function that is equal to $v$ at $w=\epsilon$ and is equal to 1 at $w=1$ (see \Cref{fig:calibrate}).
The slope of this function for $w > 0$ is $\frac{\partial\, \Calibrate(w)}{\partial w} = \frac{1 - v}{1 - \epsilon} \le \frac{16}{15} \cdot (1-v)$ for $\epsilon \le \frac{1}{16}$.\footnote{The choice of $\frac{1}{16}$ is mostly arbitrary and does not affect the final complexity.
However, note that, according to the estimate from \Cref{sec:probcoin-proof-without-calibration}, we always need to run the approximate agreement for at least $3 + \log_2 n$ rounds. Hence, it is safe to assume that $\epsilon \le \frac{1}{2^4} = \frac{1}{16}$.}


\atrev{Our goal is to reduce the problem to the case without weight calibration, but with smaller $\epsilon$.
In other words, by applying weight calibration, we manage to significantly shrink the disagreement of correct processes on the weights ``for free'', i.e., without running extra rounds of Approximate Agreement.
To this end, we define $\newEpsilon = \frac{16}{15} \epsilon (1-v)$, so that $\forall w > 0: \Calibrate(w + \epsilon) - \Calibrate(w) \le \newEpsilon$.}


However, if\atadd{, for some coordinate,} the weight received\atadd{ from {\BAA}} by one process is $\epsilon$ and the weight received by another process is $0$, then the range after the calibration will only grow\atadd{ (i.e., $\Calibrate(\epsilon) - \Calibrate(0) = v > \epsilon$)}.
Hence, we need to show that, with high probability, this discontinuity does not affect the results of the protocol.

We use the same notation as in \Cref{sec:probcoin-proof-without-calibration}: 
$T_i$ denotes the ticket of process $i$;
$w_i^j$ denotes the weight of process $i$ received by process $j$ from the approximate agreement protocol; 
$w_i^M = \max\limits_j\{w_i^j\}$ and $w_i^m = \min\limits_j\{w_i^j\}$; 
$[w_1\first, \dots, w_n\first]$ denote the weights received by the correct process that was the first to complete approximate agreement;
$X_i^j = w_i^j \cdot T_i$, $X_i\first = w_i\first \cdot T_i$, $X_i^M = w_i^M \cdot T_i$, and $X_i^m = w_i^m \cdot T_i$.
Finally, $i_M = \arg\max\limits_{i} X_i\first$ and subscript $M$ is used as a replacement for $i_M$.
I.e., $T_M = T_{i_M}$, $X_M\first = X_{i_M}\first$, $X_M^M = X_{i_M}^M$, $X_M^m = X_{i_M}^m$, and so on.

\atadd{First, let us state two useful lemmas.}

\begin{lemma} \label{lem:probcoin:calibration-failure}
    With probability at least $1 - \frac{Q}{2}$, $X_M^m > v$.
\end{lemma}

\begin{lemma} \label{lem:probcoin:calibration-success}
    With the weight calibration applied, $P[\text{failure} \mid X_M^m > v] \le P[\exists i \neq i_M: X_i\first \ge X_M\first - 2\newEpsilon]$.
\end{lemma}

\atadd{Intuitively, \Cref{lem:probcoin:calibration-failure} shows that, with high probability, a certain good event happens (namely, the discontinuity of the calibration function does not affect the outcome),
while \Cref{lem:probcoin:calibration-success} shows that, in this good case, the probability of failure is similar to that of non-calibrated algorithm with disagreement on weights equal to $\newEpsilon$.}

\atadd{Using these two lemmas, we now can prove the main theorem:}
\begin{theorem}
    With weight calibration, running the approximate agreement for $5 + \lceil\log_2(1/Q)\rceil + \lceil\log_2(\log_2(1/Q))\rceil$ rounds is sufficient to guarantee that the probability of not having agreement is at most $Q$.
\end{theorem}
\begin{proof}
    As we show in \Cref{lem:probcoin:calibration-failure},
    with probability at least $1 - \frac{Q}{2}$,
    $X_M^m$ is greater than $v$.
    %
    %
    Then, in \Cref{lem:probcoin:calibration-success}, we show that
    $P[\text{failure} \mid X_M^m > v] \le P[\exists i \neq i_M: X_i\first \ge X_M\first - 2\newEpsilon]$.
    By applying the same transformations as in \Cref{eq:probably-common-coin-main},
    we get 
    $P[\text{failure} \mid X_M^m > v]
    \le 8n\newEpsilon \le 8n\left( \frac{16}{15}\epsilon (1-v) \right) 
    \le \frac{128}{15} n \epsilon (1-v)
    =   \frac{128}{15} n \epsilon \left( \frac{\ln(2/Q)}{2N/3} \right)
    \le 9 \epsilon \log_2(2/Q)$.

    By running Approximate Agreement for $\lceil\log_2(9) + \log_2(2/Q) + \log_2(\log_2(2/Q))\rceil \le 5 + \lceil\log_2(1/Q) + \log_2(\log_2(1/Q))\rceil$ rounds, we can guarantee that $P[\text{failure} \mid X_M^m > v] \le \frac{Q}{2}$.
    
    Finally, we can compute $P[\text{failure}]$ as $P[X_M^m > v] \cdot P[\text{failure} \mid X_M^m > v] + P[X_M^m \le v] \cdot P[\text{failure} \mid X_M^m \le v]
    \le 1 \cdot P[\text{failure} \mid X_M^m > v] + P[X_M^m > v] \cdot 1
    \le \frac{Q}{2} + \frac{Q}{2} \le Q$.
\end{proof}

\atadd{Finally, let us prove the two lemmas.}
\begin{proof}[Proof of \Cref{lem:probcoin:calibration-failure}]
    Let $\CommonCore$ denote the set of processes in the common core.
    It is sufficient to show that, with probability at least $1 - \frac{Q}{2}$, there is a process $i \in \CommonCore$ such that $X_i^m = T_i > v$.
    
    Note that, by the \emph{Binding Common Core} property of Verifiable Gather and by the \emph{Unpredictability} property of $\RSD$, the tickets are unpredictable to the adversary until the common core is fixed.
    Hence $\forall i \in \CommonCore$: $T_i$ are mutually independent, and $T_i \sim U(0, 1)$.
    
    $P[\exists i \in \CommonCore: T_i > v] = 1 - P[\forall i \in \CommonCore: T_i \le v]$
    
    $P[\forall i \in \CommonCore: T_i \le v]
    = v^{|\CommonCore|}
    \le v^{2n/3}
    = \left( 1 - \frac{\ln(2/Q)}{2n/3} \right)^{2n/3}
    \le  \frac{Q}{2}$
    
    The final inequality is a special case of a more general well-known inequality (see, for example,~\cite[Appendix~B]{randomized-algorithms-textbook}):
    $\forall a > 0, x > a: \left( 1 - \frac{a}{x} \right)^{x} < e^{-a}$.
    To see why it holds, note that $\left( 1 - \frac{a}{x} \right)^{x} \xrightarrow[x \to \infty]{} e^{-a}$ and that $\left( 1 - \frac{a}{x} \right)^{x}$ increases monotonically when $x \ge a$ and $a > 0$.
\end{proof}

\begin{proof}[Proof of \Cref{lem:probcoin:calibration-success}]
    For convenience, we shall prove the equivalent statement that 
    $P[\text{success} \mid X_M^m > v] \ge P[\forall i \neq i_M: X_i\first < X_M\first - 2\newEpsilon)]$.

    As discussed in \Cref{sec:probcoin-proof-without-calibration},
    $\PSuccess \le P[\forall i \neq i_M: X_i^M < X_M^m]$.
    Hence, $P[\text{success} \mid X_M^m > v] \ge P[\forall i \neq i_M: X_i^M < X_M^m \mid X_M^m > v]$.

    It is sufficient to show that, under the condition that $X_M^m > v$, for any $i \neq i_M$: $X_i\first + \newEpsilon < X_M\first - \newEpsilon$ implies $X_i^M < X_M^m$.
    Indeed, consider two cases:
    \begin{enumerate}
        \item $w_i\first = 0:$ hence, $w_i^M \le \epsilon$ and $X_i^M = \Calibrate(w_i^M) \cdot T_i \le \Calibrate(w_i^M) \le v < X_M^m$.
        So the right hand side of the implication is always true and the implication holds trivially;

        \item $w_i\first > 0:$ hence, 
        $X_i^M 
        =    \Calibrate(w_i^M) \cdot T_i 
        \le  \Calibrate(w_i\first + \epsilon) \cdot T_i 
        \le  (\Calibrate(w_i\first) + \newEpsilon) \cdot T_i
        \le  \Calibrate(w_i\first) \cdot T_i + \newEpsilon
        =    X_i\first + \newEpsilon$.
        Similarly, $X_M^m \ge X_M\first - \newEpsilon$.
        Finally, if $X_i\first + \newEpsilon < X_M\first - \newEpsilon$, 
        then we have 
        $X_i^M
        \le  X_i\first + \newEpsilon 
        <    X_M\first - \newEpsilon
        \le  X_M^m$.
    \end{enumerate}
\end{proof}

\section{Random Secret Draw implementations}
\label{sec:rsd-impl}

\subsection{Random Secret Draw for adaptive adversary}
In RSD every process begins by generating $n$ random numbers that are assigned to each of the processes in the system and then secret sharing these numbers (\Cref{line:arsd:gen,line:arsd:ssh}. 
Once a process $i$ asserts that a value that was assigned to it was successfully secret share, it can include the source of this value to a set $S$ (\Cref{line:arsd:incs}). Once the number of processes in $S$ is $n-f$, $i$ can reliably broadcast the set it built, making it the choice of secrets to be used to form its value.

When a message containing another process' set of secrets to be used is delivered (\Cref{line:arsd:assigned}), a $\ValueAssigned$ event is triggered. After the underlying application executes $\EnableRetrieve$ (\Cref{line:arsd:enretrieve}), the correct processes enable the retrieval of the secrets chosen by every process who has already been assigned a value.

Finally, processes compute the values assigned to any user by retrieving all the secrets it has selected and summing then up in $\bmod \Domain$ (\Cref{line:arsd:sum}).

\begin{algorithm}[!htb]
    \begin{smartalgorithmic}[1]
        \DeclareLocalVariable{assignedValues}
        \DeclareLocalVariable{assignedSources}
        \DeclareLocalVariable{retrieveEnabled}


        \DeclareLocalMessageType[mRSDEnable]{RSDEnableRetrieve}
    
        \Parameters{domain size $\Domain$}

        \algspace
        \DistributedObjects
             \State $\forall j,k \in \AllProc: \AVSS_j^k$ -- instance of AVSS with leader $j$ used to generate secret for process $k$ and domain size $\Domain$ 
             \State $\forall j \in \AllProc: \BRB_j$ -- instance of BRB with leader $j$
        \EndDistributedObjects
        
        \algspace    
        \ProcessState
            \State $\assignedSources$ -- map from process ids to set of processes
            \State $\assignedValues$ -- map from process ids to numbers
            \State $\retrieveEnabled$ -- Boolean that allows values to be retrieved
        \EndProcessState
    
        \algspace
        \Operation{\StartRSD}{}
             \State $\retrieveEnabled = \False$
             \State $\forall j \in \AllProc: x_i^j = \RandomInt(\Domain)$ \label{line:arsd:gen}
             \State $S = \varnothing$
             \State $\AVSS_i^j.\ShareSecret(x_i^j)$ \label{line:arsd:ssh}
        \EndOperation  
        
        \algspace
        \Upon{$\AVSS_j.\SharingComplete()$}
            \State $S = S \cup j$ \label{line:arsd:incs}
        \EndHandler
        
        \algspace
        \Upon{$|S| = n-f$}
            \State $\BRB_i.\Broadcast(S)$
        \EndHandler
        
        \algspace
        \Upon{$\BRB_j.\Deliver(S_j) \wedge \forall k \in S_j: \AVSS_k.\SharingComplete()$} \label{line:arsd:assigned}
            \State $\assignedSources[j] = S_j$
            \State \TriggerEvent{$\ValueAssigned(j)$}
            \If{$\retrieveEnabled$} $\forall k \in S_j: \APVSS_k^j.\EnableRetrieve()$ \EndIf
        \EndHandler
        
        \algspace
        \Operation{\EnableRetrieve}{} \label{line:arsd:enretrieve}
            \State $\retrieveEnabled = \True$
            \State $\forall j \in \AllProc \wedge \ValueAssigned(j), k \in \assignedSources[j]: \APVSS_k^j.\EnableRetrieve()$
        \EndOperation
       
        \algspace
        \Operation{\RetrieveValues}{$S$} \Returns{a map from process ids to random numbers}
            \State \WaitUntil $\forall j \in S: \ValueAssigned_j()$
            \LineComment{Note that this line may not terminate if for some $j \in S$ secret sharing has not been complete}
            \LineComment{If $\APVSS_j.\Retrieve()$ has already been invoked in the past, we assume that it returns a cached value}
            \State $\forall j \in S: \assignedValues[j] = \left(\sum_{k \in \assignedSources[j]} \AVSS_k^j.\Retrieve()\right) \bmod \Domain$
            \State \Return $\assignedValues$ \label{line:arsd:sum}
        \EndOperation
    \end{smartalgorithmic}

    \caption{Random Secret Draw from AVSS, code for process $i$}
    \label{alg:arsd}
\end{algorithm}

\subsection{Cubic Random Secret Draw for static adversary}

\subsubsection{Aggregatable Publicly Verifiable Secret Sharing} \label{subsubsec:apvss}

%
The RSD implementation that we use relies on \emph{Aggregatable Publicly Verifiable Secret Sharing (APVSS)}~\cite{gurkan2021aggregatable} in order to achieve cubic communication complexity.

Unlike in AVSS, where a process shares a secret of its own choice, in APVSS, the secret is an aggregation of values generated by $n-f$ processes.
In the original presentation of~\cite{gurkan2021aggregatable}, these values are shares of a private key that is being generated.
In our case, these values are just random numbers and the aggregation is simply a sum modulo the domain size $\Domain$.
Moreover, the process itself cannot know the secret until a threshold of processes are ready to collectively reconstruct it.

For the purposes of this paper, we shall consider a restricted interface of APVSS:
\begin{objectinterface}
    \item[$\APVSS_i.\CreateScriptShare_j(x)$:] can be invoked by process $j$ to generate a \emph{script share} for process $i$;

    \item[$\APVSS_i.\VerifyScriptShare(j, \scriptShare)$:] returns $\True$ if $\scriptShare$ was generated by process $j$ invoking $\APVSS_i.\CreateScriptShare_j$;
    
    
    
    \item[$\APVSS_i.\ShareSecret(\scriptShares)$:] can be invoked by process $i$ to secret-share an aggregation of the secrets that were used to create $\scriptShares$, provided $n-f$ script shares generated by different processes;

    \item[$\APVSS_i.\SharingComplete()$, $\APVSS_i.\EnableRetrieve()$, and $\APVSS_i.\Retrieve()$:] \hfill\\ analogous to $\AVSS_i$.
\end{objectinterface}

In practice, the APVSS scheme of~\cite{gurkan2021aggregatable} provides more general functionality (e.g., it allows combining fewer than $n-f$ script shares).
However, for the Random Secret Draw protocol, this simplified interface is sufficient.

APVSS satisfies, {\pAVSSValidity}, {\pNotificationTotality}, {\pRetrieveTermination}, and {\pAVSSBinding} properties of AVSS.
Compared to AVSS, APVSS provides secrecy regardless of whether the process itself is correct:

\begin{properties}
    \item[\pAPVSSSecrecy:] if no correct process invoked $\AVSS_i.\EnableRetrieve()$, then the adversary has no information about the secret shared by $i$.
\end{properties}

\subsubsection{Pseudocode}

Similarly to the $\AVSS$ version, the $\APVSS$ version of $\RSD$ begins with every process generating $n$ random numbers to be assigned to the other participants. This time, instead of secret sharing these numbers, a scriptshare is sent which allows the recipient to combine them.

Using APVSS, every process is able to generate a secret which is a sum in modulo D of $n-f$ valid secret shares assigned to them by other processes (lines~\ref{line:rsd:aggregstart} to~\ref{line:rsd:aggregend}).

Once the secret is shared, an event $\ValueAssigned$ is generated (line~\ref{line:rsd:assigned}). This event can be treated by the application that runs RSD so that it enables the retrieval of secrets once enough values are assigned (line~\ref{line:rsd:enable}). The application must then create a set of processes $S$ for which it desires to know the random numbers of and then execute $\RetrieveValues$ which returns these values after the secrets for every processes in $S$ are assigned (lines~\ref{line:rsd:retstart} to~\ref{line:rsd:retend}). As before, only after a process gets assigned a value the secret retrieval is enabled.

\begin{algorithm}[!htbp]
    \begin{smartalgorithmic}[1]
        \DeclareLocalVariable{apvssScriptShares}
        \DeclareLocalVariable{assignedValues}
        \DeclareLocalVariable{retrieveEnabled}
        
        \DeclareLocalFunction{GetAssignedValue}

        \DeclareLocalMessageType[mApvssShare]{ApvssShare}
        \DeclareLocalMessageType[mApvssScript]{ApvssScript}
        \DeclareLocalMessageType[mSecretShare]{SecretShare}
    
        \Parameters{domain size $\Domain$}

        \algspace
        \DistributedObjects
            \State $\forall j \in \AllProc: \APVSS_j$ -- instance of APVSS with leader $j$ and domain size $\Domain$ 
        \EndDistributedObjects

        \algspace    
        \ProcessState
            \State $\apvssScriptShares$ -- set of pairs $(j, \scriptShare)$
            \State $\assignedValues$ -- map from process ids to numbers
            \State $\retrieveEnabled$ -- Boolean indicating processes may allow for secret retrieval
        \EndProcessState
    
        \algspace
        \Operation{\StartRSD}{}
            \State $\retrieveEnabled = \False$
            \State $\forall j \in \AllProc$: \Send{$\mApvssShare, \APVSS_j.\CreateScriptShare_i(\RandomInt(\Domain))$}{process $j$} 
        \EndOperation  
        
        \algspace
        \UponReceiving{$\mApvssShare, \scriptShare$}{process $j$}
            \If {not $\APVSS_i.\VerifyScriptShare(j, \scriptShare)$} \Return \Comment{ignore invalid message} \EndIf \label{line:rsd:aggregstart}
            \State $\apvssScriptShares = \apvssScriptShares \cup (j, \scriptShare)$
            \If {$|\apvssScriptShares| = n-f$}
                \LineComment{The aggregated secret is a sum modulo $D$ of the $n-f$ original secrets}
                \State $\APVSS_i.\ShareSecret(\apvssScriptShares)$
            \EndIf
        \EndUponReceiving \label{line:rsd:aggregend}

        \algspace
        \Upon{event $\APVSS_j.\SharingComplete()$}
            \State \TriggerEvent{$\ValueAssigned(j)$} \label{line:rsd:assigned}
            \If{$\retrieveEnabled$} $\APVSS_j.\EnableRetrieve()$ \EndIf \label{line:rsd:enable}
        \EndHandler
        
        \algspace
        \Operation{\EnableRetrieve}{}
            \State $\retrieveEnabled = \True$
            \State $\forall j \in \AllProc \wedge \ValueAssigned(j) \APVSS_j.\EnableRetrieve()$
        \EndOperation
        
        \algspace
        \Operation{\RetrieveValues}{$S$} \Returns{a map from process ids to random numbers} \label{line:rsd:retstart}
            \State \WaitUntil $\forall j \in S: \ValueAssigned_j()$
            \LineComment{Note that this line may not terminate if for some $j \in S$ secret sharing has not been complete}
            \LineComment{If $\APVSS_j.\Retrieve()$ has already been invoked in the past, we assume that it returns a cached value}
            \State $\assignedValues = \text{map} [j \mapsto \APVSS_j.\Retrieve() \mid \forall j \in S]$
            \State \Return $\assignedValues$
        \EndOperation \label{line:rsd:retend}
    \end{smartalgorithmic}

    \caption{Random Secret Draw from APVSS, code for process $i$}
    \label{alg:rsd}
\end{algorithm}

\newpage
\section{\atremove{Modern }Pseudocode for Canetti and Rabin~\cite{CanettiRabin93} common coin}
\label{app:cr93}

Using our building blocks and presentation, the pseudocode for the coin proposed by Canetti and Rabin becomes \atreplace{extremely simplified}{very simple}. It consists of every process executing random secret draw and then gathering a set of processes who get assigned values. If among the corresponding gathered values there is a $0$, then a $0$ is decided, else the processes decide $1$.

\atremove[Let's remove this paragraph out of respect to Canetti and Rabin :)\@. They are not our competitors anyway.]{Notice that if there is a zero in the common core, then the processes are guaranteed to decide the same value. Moreover, this coin has a\atremove{ very} skewed distribution, unlike ours which guarantees uniform distribution of the outputs, as well as multi valued randomness.}

\atrev{Here, the advantages of a modular approach can be seen: first, the algorithm becomes clearer, and, second, later improvements in building blocks lead to an improvement of the algorithm as a whole.}

\begin{algorithm}[!htb]
    \begin{smartalgorithmic}[1]
        \DeclareLocalVariable{candidates}
        
        \DeclareLocalVariable{tickets}
        \DeclareLocalVariable{values}
        \DeclareLocalVariable{winner}
        
        \algspace
        
        \DistributedObjects
            \State $\RSD$ -- instance of Random Secret Draw with \atrev{domain size $\lceil 0.87n \rceil$}
            \State $\Gather$ -- instance of Gather (verifiability is not necessary for this protocol)
        \EndDistributedObjects
        
        \algspace
        \Function{\GatherAccept}{$j$} \Returns{\Boolean}
            \State \Return $\True$ iff $\RSD.\ValueAssigned(j)$
        \EndFunction
            
        \algspace
        \Operation{\Toss}{} \Returns{\Integer}
            \State $\RSD.\StartRSD()$ 
            \State $\Gather.\StartGather(\GatherAccept)$ 
        \State \atrev{\WaitFor{event $\Gather.\DeliverSet(S)$}}
        \State $\RSD.\EnableRetrieve()$ 
        \State $\values = \ValueDraw.\RetrieveValues(\candidates)$
        \If{$0 \in \values$}
            \State \Return $0$
        \Else{\State \Return $1$}
        \EndIf
        \EndOperation
    \end{smartalgorithmic}
    
    \caption{Canetti and Rabin Common Coin, code for process $i$}
    \label{alg:cr93}
\end{algorithm}

\section{Concrete Code for Intersecting Random Subsets}
\label{app:committee-elections-concrete-code}

The following recursive formula produces a code (i.e., a list of code words) that satisfy the definition in \Cref{sec:committee-elections}:

\begin{align*}
    \forall n \geq 0~\text{and}~m \in [0..n]: C_{n,m} = 
    \begin{cases}
        [\text{``}~"]                                   & \text{if}~0 = m = n \\
        [\text{``}000\dots000"]                         & \text{if}~0 = m < n \\
        [\text{``}111\dots111"]                         & \text{if}~0 < m = n \\
        0{\Vert}C_{n-1,m}, \mathrm{reverse}(1{\Vert}C_{n,m-1})  & \text{otherwise} \\
    \end{cases}
\end{align*}

For example, here is $C_{5,2}$: [00011, 00110, 00101, 01100, 01010, 01001, 11000, 10100, 10010, 10001].

In order to avoid computing all $\binom{n}{m}$ code words, when $n$ and $m$ are large, we can use the following recursive formula to efficiently (with $O(poly(n))$ operations) find a code word by its index:

\begin{gather*}
    \forall n \geq 0~\text{and}~m \in [0..n], i \in [0..\binom{n}{m}{-}1]: \\
    C_{n,m}[i] = 
    \begin{cases}
        \text{``}~"             & \text{if}~0 = m = n \\
        \text{``}000\dots000"   & \text{if}~0 = m < n \\
        \text{``}111\dots111"   & \text{if}~0 < m = n \\
        0{\Vert}C_{n-1,m}[i]    & \text{if}~0 < m < n, i < {\binom{n-1}{m}}\\
        1{\Vert}C_{n-1,m-1}[\binom{n-1}{m-1} - \left(i - \binom{n-1}{m}\right) - 1]
                        & \text{if}~0 < m < n, i \ge \binom{n-1}{m}
    \end{cases}
\end{gather*}


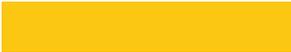
\end{document}